\documentclass[draftcls,onecolumn,12pt]{IEEEtran}
\usepackage{xspace,amsthm,amsmath,amssymb,amsfonts,epsfig,syntonly}
\usepackage{cite,bm,color,url,textcomp,balance}
\usepackage{epstopdf}
\usepackage{empheq}

\usepackage{algorithm}
\usepackage{algorithmicx}
\usepackage{algpseudocode}

\usepackage{graphicx}{\large}
\usepackage{subfigure}

\usepackage{amsmath,amssymb,amsthm}
\usepackage{bigints}
\usepackage{stfloats,siunitx}

\newcommand*{\blue}{\color{black}}

\newtheorem{lemma}{Lemma}

\newtheorem{proposition}{Proposition}

\newtheorem{remark}{Remark}

\long\def\symbolfootnote[#1]#2{\begingroup
\def\thefootnote{\fnsymbol{footnote}}
\footnote[#1]{#2}\endgroup}

\psfull

\allowdisplaybreaks[4]

\IEEEoverridecommandlockouts

\begin{document}
\title{Energy Management and Trajectory Optimization for UAV-Enabled Legitimate Monitoring Systems}

\author{Shuyan Hu, Qingqing Wu,~\IEEEmembership{Member, IEEE},
and Xin Wang,~\IEEEmembership{Senior Member, IEEE}
\thanks{Work in this paper was supported by the National Natural Science Foundation of China under Grant No. 61671154,
the Shanghai Science Foundation under Grant No. 18ZR1402700,
and the Innovation Program of Shanghai Municipal Education Commission.}
\thanks{S. Hu is with the State Key Laboratory of ASIC and System, the School of Information Science and Technology, Fudan University, Shanghai 200433, China (e-mail: syhu14@fudan.edu.cn).

Q. Wu is with the Department of Electrical and Computer Engineering, National University of Singapore, Singapore 119077
(e-mail: elewuqq@nus.edu.sg).

X. Wang is with the State Key Laboratory of ASIC and System, the Shanghai Institute for Advanced Communication and Data Science, the Department of Communication Science and Engineering, Fudan University, Shanghai 200433, China
(e-mail: xwang11@fudan.edu.cn).
}}

\maketitle

\begin{abstract}
Thanks to their quick placement and high flexibility, unmanned aerial vehicles (UAVs) can be very useful in the
{\blue current and future wireless communication systems.}
With a growing number of smart devices and infrastructure-free communication networks,
it is necessary to legitimately monitor these networks to prevent crimes.
In this paper, a novel framework is proposed to exploit the flexibility of the UAV for legitimate monitoring
via joint trajectory design and energy management.
The system includes a suspicious transmission link with a terrestrial transmitter and a terrestrial receiver,
and a UAV to monitor the suspicious link.
The UAV can adjust its positions and send jamming signal to the suspicious receiver to ensure successful eavesdropping.
Based on this model, we first develop an approach to minimize the overall jamming energy consumption of the UAV.
Building on a judicious (re-)formulation, an alternating optimization approach is developed
to compute a locally optimal solution in polynomial time.
Furthermore, we model and include the propulsion power to minimize the overall energy consumption of the UAV.
Leveraging the successive convex approximation method,
an effective iterative approach is developed to find a feasible solution fulfilling the Karush-Kuhn-Tucker (KKT) conditions.
Extensive numerical results are provided to verify the merits of the proposed schemes.

\end{abstract}


\section{Introduction}

Featuring high flexibility, swift deployment, and wide coverage,
unmanned aerial vehicles (UAVs) have been extensively applied to activities such as
search and rescue in disaster areas, inspection of landscapes, and surveillance of forrest fires.
Recently, UAVs have found many use cases in wireless communication networks as cost-effective and on-demand aerial wireless platforms for areas without cellular coverage~\cite{Zeng, wu19, zeng19},
{\blue or as flying mobile users within a cellular network~\cite{survey1, survey2}.
Cellular-connected UAVs can enhance connectivity, coverage, flexibility and reliability
of wireless communication networks~\cite{survey1, survey2}.
}
The UAVs are anticipated to engage significantly in the fifth-generation (5G) and beyond 5G (B5G) wireless networks,
and provide new services such as real-time image transmission \cite{mot17}, caching and multicasting \cite{xiaoli, zeng18},
data dissemination or collection \cite{Zeng, wu18mar, wu18dec}, mobile relaying and edge computing \cite{kli16, zeng16, yang19},
and wireless power transfer \cite{moza17, xu18, wu19}.

As the applications of the internet-of-things (IoT) continue to expand in the
{\blue current and future wireless networks,}
many infrastructure-free wireless links (such as bluetooth, Wi-Fi, and UAV-enabled transmission) have been established
to support communications among IoT devices.
Yet, these convenient networks can be abused for crimes and terrorism, if in the wrong hands.
Therefore, it is necessary for authorized parties to surveil these suspicious communication links
(see \cite{huang18, zeng2016, jie2018, xu17, xu17surv, cai17, hu17, haiquan19, moon19, wu19oct}).
Optimization metrics for legitimate monitoring typically focused on maximizing the eavesdropping rate
or the non-outage probability \cite{huang18}.
Spoofing schemes were proposed for a malicious transmission link
to maximize the eavesdropping rate \cite{zeng2016}, or to intervene and change the communicated data \cite{jie2018}.
For a suspicious communication link in \cite{xu17}, the largest achievable monitoring non-outage probability
and comparative intercepting rate were obtained under delay-sensitive and delay-tolerant scenarios, respectively.
Proactive jamming schemes were developed to maximize the average monitoring rate for
multi-input multi-output (MIMO) channels \cite{cai17}, relay networks \cite{hu17}, UAV-aided links \cite{haiquan19},
and with a deep-learning approach \cite{moon19}.

Most existing works considered a fixed ground node (GN) as the legitimate monitor,
whose channel typically suffers from severe large-scale path loss and small-scale fading.
Yet, the UAV-enabled monitor can enjoy high-probability of line-of-sight (LoS) channels
as its flying altitude rises.
It is therefore easier for the UAV to obtain its channel gains with the GNs if their locations are known.
Thanks to its flexibility, 
the UAV can dynamically adjust its positions
for better eavesdropping rate, e.g., by flying closer to the suspicious transmitter.
Categorized by the power sources, 
there are two types of UAVs, namely, the tethered UAVs and the untethered UAVs \cite{qq18}.
A tethered UAV is linked with a ground control platform,
and is powered stably through a cable or a wire.
The lack of mobility has constrained tethered UAVs to a targeted area only \cite{yali16, yang17, lyu17, alze17}.
In particular, the horizontal positions of the UAVs were optimized in \cite{lyu17} to cover a set of GNs with the least possible
number of UAVs.
The optimal three-dimensional (3D) deployment scheme of a UAV was developed in \cite{alze17} to
cover as many GNs as possible with a minimum transmit power budget.

By contrast, untethered UAVs are powered by laser-beam, on-board battery, and/or solar panel.
They can fly freely and enjoy full mobility in wide 3D space.
Communication throughput was maximized for a laser-powered UAV in \cite{ouyang18}.
In spite of their flexibility, the battery-powered UAVs have to revisit their home base repeatedly to refill their batteries during operations,
due to the limited capacity of on-board batteries \cite{derrick}.
The optimal UAV trajectory and power management schemes were developed in \cite{zeng16}
to obtain the largest achievable data rate of a relaying system,
and in \cite{zeng18} to minimize the data dissemination time of a multicasting system.
Since solar panels at the UAVs can harness and convert energy to electric power, supporting long endurance flights,
solar-powered UAVs have also received great research interests.
The optimal 3D trajectory optimization and resource assignment for a solar-powered UAV-aided communication system
were developed in \cite{derrick} to achieve the largest overall data rate in a fixed time horizon.

Apart from transmit power, the UAVs consume additional propulsion power to support hovering and moving activities.
As a result, the energy management for UAV-enabled communications noticeably differs
from that in current systems on the ground.
The largest value of energy efficiency in bits/Joule was obtained in \cite{zeng17} for a fixed-wing UAV via trajectory optimization.
Total (including communication and propulsion) energy usage of a rotary-wing UAV was minimized in \cite{zyong}
to satisfy the throughput requirement of each GN.

In this paper, we propose a simple model for a rotary-wing UAV enabled monitoring system.
The suspicious transmission link on the ground consists one source (transmitting) node S and one destination (receiving) node D.
When the UAV's channel condition is worse than that of node D,
the UAV sends jamming signal to the latter as noise to degrade its channel for successful eavesdropping.
The total jamming energy consumption of the UAV is minimized in a finite period via joint trajectory optimization and power allocation,
based on the assumption of successful eavesdropping at each slot.
By judicious reformulation, we transform this non-convex optimization task into two separable subproblems,
each of which is convex when the other set of variables are fixed.
The alternating optimization method is leveraged to develop an efficient approach that is ensured to
converge to a locally optimal solution.
Based on such a solution, some useful insights are also drawn on the changing patterns of the UAV's trajectory and jamming policy.
To achieve energy-efficient UAV operations in practice,
we further consider a solar-powered rotary-wing UAV enabled monitoring system by including
the propulsion power consumption besides the jamming power.
Capitalizing on the successive convex approximation (SCA) method,
an efficient iterative approach is put forth to find a feasible solution fulfilling the Karush-Kuhn-Tucker (KKT) conditions.
Numerical results demonstrate that with UAV trajectory optimization, the overall energy consumption can be greatly suppressed.

The rest of the paper is organized as follows.
Section \ref{sec:model} describes the system models.
Section \ref{sec:jam} develops an approach to the UAV trajectory design and jamming energy minimization,
while Section \ref{sec.energy} addresses the trajectory design and total (jamming and propulsion) energy management
for a solar-powered UAV.
Numerical results are provided in Section \ref{sec.sim}.
The paper is concluded in Section~\ref{sec.con}.

\begin{figure}[t]
\centering
\includegraphics[width=0.6\textwidth]{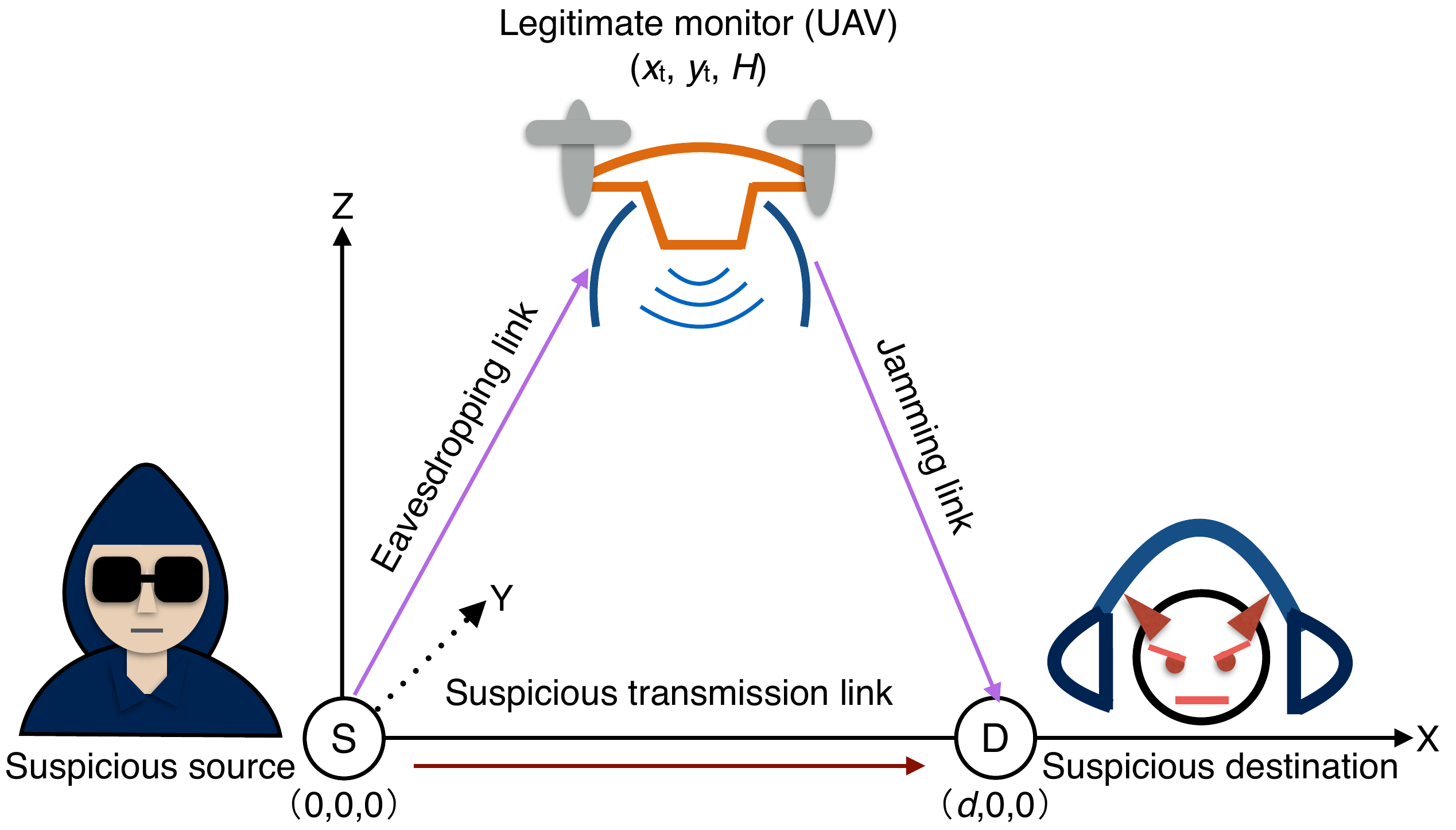}
\caption{A UAV-enabled legitimate monitoring system.}
\label{syst}
\end{figure}

\section{System Models}\label{sec:model}
Consider a point-to-point, frequency non-selective wireless communication link from a suspicious source node S
to a suspicious destination node D which are geographically set apart by $d$ meters on the ground.
An untethered UAV, traveling at a fixed altitude of $H$ meters,
serves as the legitimate monitor to eavesdrop this link;
see Fig. \ref{syst}.\footnote{\blue{Although design freedom can be increased by further optimizing UAV's altitude,
energy consumption as well as risks of instability and collision will rise. Therefore, rather than frequently adjusting altitude,
it may be better for the UAV to fly at a fixed altitude and avoid vertical movement due to airspace regulation,
collision avoidance, energy saving and safety concerns.}}
The UAV can move forward horizontally or hover in the air.
It can travel in the vicinity above the GNs to improve its eavesdropping performance.
Suspicious nodes S and D have one antenna each, and the UAV operates with two antennas,
one for monitoring and intercepting information from the S-D link (receiving)
and the other for sending jamming signals to node D (transmitting).
Therefore the UAV can perform in a full-duplex state to jam and monitor simultaneously.
Since its initial and final locations are given,
the UAV's channel power gain can be worse than that of D at certain time.
In this case, the UAV sends jamming signal to the latter as noise to degrade its channel for successful eavesdropping.
We assume that the UAV can completely annul its self-interference from the transmitting antenna to the receiving antenna
by adopting state-of-the-art analog and digital self-interference cancelation schemes \cite{xu17}.

\subsection{UAV Mobility Model}\label{sec:uavmob}
Without loss of generality, we consider a 3D Cartesian coordinate system with nodes S and D
located at $(0,0,0)$ and $(d,0,0)$, respectively.
The UAV is deployed for the monitoring mission in a finite scheduling horizon of $T$ seconds.
We split the period $T$ into $T_w$ time slots given by ${\cal T}:= \{1, \ldots, T_w \}$;
the duration of each slot is the same as $\delta$.
The slot length is selected to be short enough so that the UAV can be treated as static within each slot.
Consequently, the time-varying coordinates of the UAV are given by $(x_t,y_t,H), \forall t\in{\cal T}$,
with $x_t$ and $y_t$ being the UAV's x- and y-coordinates over time, respectively.
The initial and final locations of the legitimate monitor are pre-defined and given by $(x_0,y_0,H)$, and $(x_T,y_T,H)$, respectively.
The minimum traveling distance for the UAV to finish during the scheduling horizon $T$ is thereby
$d_{\min}=\sqrt{(x_T-x_0)^2+(y_T-y_0)^2}$.
Given the maximum speed of the UAV $\tilde V_m$, we let $\tilde V_m \geq d_{\min}/T$ so that
at least one feasible trajectory can be found from the UAV's initial to final locations.\footnote{\blue{By considering the time for acceleration, the proposed maximum speed $\tilde V_m$ may be infeasible. However, in practice, the acceleration time could be very short and thus reasonably ignored, especially when the total flying period or distance is sufficiently long. From this perspective, we provide a lower bound for the maximum speed.}}

Consequently, the UAV's mobile activity constraints, including its initial and final locations and speed constraints are given by~\cite{zeng16}:
\begin{subequations}
\begin{align}
(x_1-x_0)^2+(y_1-y_0)^2 &\leq V_m^2 \label{eq.mob1} \\
\quad(x_{t+1}-x_t)^2+(y_{t+1}-y_t)^2 &\leq V_m^2, ~\forall t \in {\cal T} \label{eq.mob2} \\
(x_T-x_{T-1})^2+(y_T-y_{T-1})^2 &\leq V_m^2 \label{eq.mob3}
\end{align}
\end{subequations}
where $V_m := \tilde V_m \delta$ stands for the largest traveling distance of the UAV for each slot.

{\blue
\begin{remark}\textit {(The choice of $T_w$):}
In general, $T_w$ is chosen such that the UAV can be treated as (quasi-) static within each time slot, observed from the ground.
To guarantee a certain accuracy, the ratio of the largest traveling distance within each time slot $\tilde V_m \delta$
and the UAV altitude $H$ can be restricted below a threshold, i.e., $\tilde V_m \delta/H \le \varepsilon_m$,
where $\varepsilon_m$ is the given threshold and $\delta = T/T_w$.
Then, the minimum number of time slots required for achieving the accuracy with a given $\varepsilon_m$ can be obtained as
$T_w \ge \tilde V_m T /(H \varepsilon_m)$.
The optimization gets more precise with more discretized time samples, i.e., larger value of $T_w$.
Yet, the computational complexity, given by $\mathcal{O}(T_w^{3.5})$, also increases significantly with the value of $T_w$.
Therefore, the number of time slots $T_w$ can be properly chosen in practice to balance between the accuracy and complexity~\cite{wu18mar}.
\end{remark}
}

\subsection{Communication Channel Model}

Malicious users of infrastructure-free wireless communication networks are more likely to appear in wide rural areas,
where surveillance is overlooked.
In open rural areas, the buildings and trees are sparsely distributed.
LoS channels can be dominant even for communications between GNs.
Therefore, we can suppose that the communication links between S, D, and UAV (i.e., node U) are all dominated by LoS channels,
which can facilitate analysis on the structural properties of the optimal solution.
The case with non-LoS channels will be accordingly addressed later.
We further suppose that the Doppler effect resulted from the UAV's mobile activities is completely neutralized \cite{zeng16, lin2018}.
The distance between S and D is fixed during the entire scheduling horizon, i.e., $d_{SD} = d$ meters.
Hence, the channel power gain of the suspicious link from S to D is constant and can be expressed as
\begin{equation}
h_{0}=\frac{\beta_0}{{d_{SD}}^2}=\frac{\beta_0}{d^2}
\end{equation}
where $\beta_0$ stands for the channel power at the reference distance $d_0=1$ meter.
At each slot $t$, the channel power gain from S to U for legitimate eavesdropping follows the LoS model as
\begin{equation}
h_{1}^t=\frac{\beta_0}{{d_{1}^{t}}^2}=\frac{\beta_0}{x_t^2+y_t^2+H^2}, ~~\forall t \in {\cal T}
\end{equation}
where $d_{1}^t= \sqrt{x_t^2+y_t^2+H^2}$ is the link distance between S and U at slot $t$.
Similarly, the channel power gain from U to D for jamming is
\begin{equation}
h_{2}^t=\frac{\beta_0}{{d_{2}^{t}}^2}=\frac{\beta_0}{(d-x_t)^2+y_t^2+H^2}, ~~\forall t \in {\cal T}
\end{equation}
where $d_{2}^{t}=\sqrt{(d-x_t)^2+y_t^2+H^2}$ is the separation distance between U and D at slot $t$.

Let $P_x^t$ stand for the transmit power by S at time slot $t$,
and $P_j^t$ the jamming power from U to D to interfere the channel at the suspicious receiver for a successful eavesdropping.
Clearly, the signal-to-interference-plus-noise ratio (SINR) at the suspicious receiver D is
\begin{equation}\label{snrd}
\gamma_{D}^t= \frac{h_{0}P_x^t}{h_{2}^tP_j^t+\sigma^2}, ~~\forall t \in {\cal T}
\end{equation}
where $\sigma^2$ is variance of the additive white Gaussian noise (AWGN).
On the other hand, the UAV can completely annul its self-interference from its jamming antenna to its receiving antenna.
Hence, the SINR (which in fact reduces to signal-to-noise ratio, SNR) 
of the legitimate eavesdropping channel at U is
\begin{equation}\label{snru}
\gamma_{U}^t= \frac{h_{1}^tP_x^t}{\sigma^2}, ~~\forall t \in {\cal T}.
\end{equation}

Successful eavesdropping at the UAV requires $\gamma_U^t \geq \gamma_D^t$.
The UAV can achieve this goal by dynamically adjusting its trajectory to fly close to the source node,
and/or adjusting its jamming power to reduce the channel gain of the suspicious receiver D at each time slot,
when the channel condition of the UAV is worse than that of D.\footnote{\blue{When the channel condition of the legitimate monitor (the S-U link) is better than that of the suspicious receiver
(the S-D link), eavesdropping is performed successfully without the UAV sending jamming signals to the receiver.
However, when the S-U link suffers a worse channel condition than the S-D link,
successful eavesdropping can be enabled through letting the UAV send jamming signals to the receiver to degrade its channel condition.
}}

Note that the assumption of successful eavesdropping at each time slot is tenable and non-trivial.
In fact, malicious users of infrastructure-free wireless communication networks can also develop counter-eavesdropping measures
to ensure secure transmissions on their behalf.
One important method is to transmit secret information in cipher.
In order to learn the pattern and decode the secret information, the legitimate agency treasures every bit of information.
In this case, the cipher transmitted in each time slot is of equal importance for the legitimate agency
to piece together the whole picture.
Hence, it is of paramount importance if the eavesdropper can intercept information from the suspicious link successfully in every time slot.

{\blue
\begin{remark}\textit {(Decoding the intercepted information):}
In this paper, we aim to investigate the fundamental performance limits of the physical layer approach for eavesdropping,
and thus do not consider encryption for the suspicious link, which is a higher layer technique and can be resolved as long as there are powerful computing resources. 
On the other hand, to avoid being monitored and tracked by legitimate parties, the suspicious link is very likely built, used and discarded
or destroyed in a day,
which makes the link not complete or mature enough in terms of software and hardware to preserve privacy and security.
We can thereby reasonably assume that the temporarily-established infrastructure-free suspicious link
is not vigilant against eavesdropping and does not employ
any countermeasures such as signal encryption or anti-surveillance detection.
From this perspective, the UAV can successfully decode the intercepted information from the suspicious link.
\end{remark}
}
\section{Legitimate Eavesdropping with Jamming}\label{sec:jam}

To ensure successful eavesdropping, the UAV may need to jam the transmission from S to D.
For an untethered UAV without incessant power supply, it is clear that we wish to minimize its overall jamming energy consumption.
Building on the UAV's mobile activity constraints \eqref{eq.mob1}--\eqref{eq.mob3},
together with the SINR expressions \eqref{snrd}--\eqref{snru},
the optimization task of interest can be formulated as
\begin{subequations}\label{p1}
\begin{align}
& \min_{\{P_j^t\}, \{x_t, y_t\}} \sum_{t \in {\cal T}} P_j^t \delta \label{p11}\\
&\text {s.t.} ~\frac{h_{0}P_x^t}{h_{2}^tP_j^t+\sigma^2} \leq \frac{h_{1}^tP_x^t}{\sigma^2}, ~\forall t  \label{p12}\\
&(x_1-x_0)^2+(y_1-y_0)^2\leq V_m^2 \label{p13}\\
&(x_{t+1}-x_t)^2+(y_{t+1}-y_t)^2\leq V_m^2, ~\forall t  \label{p14}\\
&(x_T-x_{T-1})^2+(y_T-y_{T-1})^2\leq V_m^2 \label{p15}\\
&P_j^t \geq 0, ~\forall t. \label{p16}
\end{align}
\end{subequations}

Here we in fact aim to pursue the optimal jamming policy and trajectory design for the UAV.
Note that the transmit power $P_x^t$ by S can be canceled from the both sides of the inequality constraints in \eqref{p12}.
This implies that the UAV does not need to know the $P_x^t$ when making its jamming and trajectory decisions.
This is of practical interest as the suspicious source is certainly reluctant to let the UAV know its transmit power value.

\subsection{Proposed Solution}
Problem \eqref{p1} is not a convex program because of the non-convex constraints in \eqref{p12};
hence, it cannot be dealt with by classic convex optimization methods.
To make the problem more tractable, we introduce two slack variables
$u_t := x_t^2 + y_t^2 + H^2$, and $w_t := (d-x_t)^2 + y_t^2 + H^2$,
and rewrite \eqref{p1} as
\begin{subequations}\label{p2}
\begin{align}
& \min_{\{P_j^t, u_t, w_t\}, \{x_t, y_t\}} \sum_{t \in {\cal T}} P_j^t \delta \label{p21}\\
&\text {s.t.} ~x_t^2 + y_t^2 + H^2 - u_t \leq 0, ~\forall t \label{p22}\\
& u_t - 2dx_t + d^2 - w_t \leq 0, ~\forall t \label{p23}\\
&\frac{u_t w_t}{d^2} - w_t - P_j^t \beta_0/ \sigma^2 \leq 0, ~\forall t  \label{p24}\\
&w_t \geq H^2, ~\forall t \label{p25}\\
&\eqref{p13} - \eqref{p16} \notag
\end{align}
\end{subequations}
where \eqref{p24} results from \eqref{p12} by the following step
\begin{equation}
\frac{h_{0}}{P_j^t \beta_0/w_t+\sigma^2} \leq \frac{\beta_0/u_t}{\sigma^2}, ~\forall t.  \label{eq.jam}
\end{equation}
Note that we change the ``='' signs to ``$\leq$'' signs in \eqref{p22} and \eqref{p23} to convexify those constraints.
It can be justified that upon obtaining the optimal solution for \eqref{p2},
constraints \eqref{p22} and \eqref{p23} should always be met with equality,
since otherwise, we can always decrease $u_t$ and $w_t$, respectively, to improve the channel condition of the corresponding eavesdropping and jamming link, leading to smaller total jamming energy consumption.
Therefore, problems \eqref{p1} and \eqref{p2} are equivalent.

Although problem \eqref{p2} is not convex, it is easy to see that the problem becomes 
convex with regard to $\{P_j^t, x_t, y_t, u_t\}$ for fixed $\{w_t \}$,
and it is also convex in $\{w_t \}$ for fixed $\{P_j^t, x_t, y_t, u_t\}$.
For this reason, we resort to the alternating optimization method (a.k.a. block coordinate descent) to solve \eqref{p2}.
The proposed algorithm is summarized in Algorithm \ref{algo:bcd}.
Since both subproblems are convex,
the globally optimal solution for each of them can be obtained by standard convex optimization solvers,
e.g., the interior point methods, in polynomial time \cite{Boyd}.
Clearly, the total jamming energy of UAV is bounded above zero.
For the proposed block coordinate descent method,
the resultant total jamming energy is decreased in each iteration.
Consequently, the proposed approach is ensured to converge to a locally optimal solution for problem \eqref{p2}.
As problems \eqref{p1} and \eqref{p2} are equivalent, a locally optimal solution for \eqref{p1} can be readily obtained.

\begin{algorithm}[t]
\caption{Alternating Optimization for Problem \eqref{p2}}
\label{algo:bcd}
\begin{algorithmic}[1]
\State {\bf Initialize} $\{P_j^t(0), x_t(0), y_t(0), u_t(0)\}$, and set initial feasible values of $\{w_t(0) \}$ for Problem \eqref{p2}.
\For {$m$ = 0, 1, 2, ...}
\State Obtain the optimal solution of $\{P_j^t(m+1), x_t(m+1), y_t(m+1), u_t(m+1)\}$ with $\{w_t(m) \}$ fixed.
\State Compute the optimal solution of $\{w_t(m+1)\}$ with $\{P_j^t(m+1), x_t(m+1), y_t(m+1), u_t(m+1)\}$ fixed. 
\State Update $m=m+1$.
\EndFor
\end{algorithmic}
\end{algorithm}

\subsection{Structural Properties}
To draw useful insights on the optimal trajectory optimization and jamming power allocation scheme, 
we analyze the structural properties of the optimal solution
for the UAV-aided eavesdropping system.

\begin{lemma}\label{lemma.free}
When the UAV is in the circular area of ${\cal A} :=\{(x_t, y_t) | \sqrt{x_t^2 + y_t^2 + H^2} \leq d, \forall t \}$,
it can eavesdrop successfully without jamming, i.e., $P_j^t=0, \forall t$.
\end{lemma}
\begin{proof}
Lemma 1 can be proven through analyzing the characteristics of the transmit and eavesdropping rate.
When the UAV is in the circular area of ${\cal A} :=\{(x_t, y_t) | \sqrt{x_t^2 + y_t^2 + H^2} \leq d, ~\forall t \}$,
the quality of the channel from S to U ($h_1^t=\beta_0/(x_t^2+y_t^2+H^2)$) is the same as or better than 
that from S to D ($h_0=\beta_0/d^2$).
It then readily follows that the UAV can eavesdrop successfully without jamming.
\end{proof}

The circular area of ${\cal A}$ can be referred to as the jamming-free area.
When the UAV is out of the range of $\cal A$, the channel quality from S to U is worse than that from S to D.
In this case, the UAV can only eavesdrop successfully by degrading the SINR of the S-D link through jamming.
The amount of the jamming power at each time slot increases with the UAV's distance to S.

Based on Lemma \ref{lemma.free}, it can be inferred that when both the initial and final locations of the UAV are inside ${\cal A}$,
the optimal jamming policy is always zero, i.e., ${P_j^t}^* = 0, ~\forall t$.
As a result, the optimization problem \eqref{p1} reduces to find a feasible trajectory within the circular area
of ${\cal A}$ with $P_j^t = 0, ~\forall t$, i.e.,
\begin{equation}\label{p1easy}
\begin{aligned}
& \text{find}~ { \{x_t, y_t\}} \\
&\text {s.t.} ~x_t^2 + y_t^2 + H^2 \leq d^2, ~\forall t  \\
&\eqref{p13}-\eqref{p15}.
\end{aligned}
\end{equation}
Since problem \eqref{p1easy} is convex, a classic convex solver can be leveraged to obtain the optimal solution,
which is not necessarily unique.

\begin{lemma}\label{lemma.time}
When the scheduling horizon $T$ is larger than the minimum traveling time of the UAV $T_{\min}= d_{\min}/ \tilde V_m$,
the UAV will first fly towards the jamming-free area, then fly to its final location.
\end{lemma}
Lemma \ref{lemma.time} is quite intuitive, as the UAV enjoys a better channel condition when it is closer to S.
Based on Lemma \ref{lemma.time}, we can further characterize the changing patterns of the UAV's jamming policy.

\begin{proposition}\label{prop.jam}
In general, the UAV's jamming power obeys the rule of first non-increasing then non-decreasing.
In some special cases, the jamming power either always non-increasing, or always non-decreasing.
\end{proposition}
\begin{proof}
When the UAV trajectory is fixed, \eqref{p1} reduces to a jamming energy minimization problem:
\begin{equation}\label{p1reduce}
\begin{aligned}
& \min_{\{P_j^t\}} \sum_{t \in {\cal T}} P_j^t  \delta \\
&\text {s.t.} ~\frac{h_{0}P_x^t}{h_{2}^tP_j^t+\sigma^2} \leq \frac{h_{1}^tP_x^t}{\sigma^2}, ~\forall t  \\
&P_j^t \geq 0, ~\forall t.
\end{aligned}
\end{equation}
For each time slot, the optimal solution of the jamming power is given by
${P_j^t}^* = \max\{0, \frac{\sigma^2}{h_2^t}(\frac{h_0}{h_1^t}-1)\}$,
where ${P_j^t}^* =0$ when the UAV is in the jamming-free area of $\cal A$,
and ${P_j^t}^* = \frac{\sigma^2}{h_2^t}(\frac{h_0}{h_1^t}-1) >0$ when the UAV is outside $\cal A$.
The latter can be rewritten into
\begin{equation}
{P_j^t}^* = \frac{\sigma^2}{\beta_0 d^2}[(d-x_t)^2+y_t^2 + H^2] [(x_t^2+y_t^2 + H^2)-d^2]
\end{equation}
where $x_t^2 + y_t^2 + H^2 \geq d^2$.
The projection of the jamming-free area on the ground is a circle centered at S (0,0), with the radius of $\sqrt{d^2-H^2}$.
To observe how ${P_j^t}^*$ changes with $x_t$ outside $\cal A$, we let $y_t^2 = d^2 - H^2$
and take the first-order partial derivative of ${P_j^t}^*$ over $x_t$:
\begin{equation}\label{eq.derix}
\partial {P_j^t}^* / \partial x_t = 4x_t^3 - 6dx_t^2 + 4d^2x_t=x_t[(2x_t-3d/2)^2+7d^2/4].
\end{equation}
Clearly, the optimal jamming power ${P_j^t}^*$ increases with $x_t$ when $x_t >0$, and decreases with it when $x_t <0$.
The same pattern can be drawn from ${P_j^t}^*$ with respect to $y_t$.
In one word, ${P_j^t}^*$ increases as the UAV flies away from S.

Now consider the following three cases.

Case i): Initial and final locations are both outside ${\cal A}$.
When the UAV's traveling time is abundant, i.e., $T>T_{\min}$, it always seeks the trajectory
that yields the least energy consumption.
Therefore, the UAV first flies towards ${\cal A}$, then to its final destination.
The jamming power experiences the process of first decreasing then increasing.
The same jamming policy applies when $T=T_{\min}$ and the line segment connecting the initial and final points
goes through ${\cal A}$.

Case ii): Initial (or final) location is inside (or outside) ${\cal A}$, or vice versa.
In the first scenario, the jamming power first decreases to zero,
then stays constant till the eavesdropping mission is accomplished.
The jamming power is always non-increasing.
If we switch the initial and final locations, the jamming power then experiences a non-decreasing process.

Case iii): Both the initial and final locations are inside ${\cal A}$.
The jamming power is always zero in this scenario.

Combining Cases i)--iii), the proposition follows.
\end{proof}

Proposition \ref{prop.jam} provides important insights on the optimal jamming policy of the UAV according to different
initial and final locations.
It shows that the UAV is willing to travel slowly inside the jamming-free area
and even take detours to reduce the jamming power consumption.
Such a strategy of the UAV is typically the consequence of minimizing the jamming energy only. 

{\blue
\begin{remark}\textit {(In and out of the jamming-free area):}
The UAV usually stays in the home base, awaiting mission assignment,
and is dispatched as a legitimate monitor once a suspicious link is detected.
As the exact location of the suspicious link is not predictable,
it is not likely that the UAV happens to be within the jamming-free area every time.
Furthermore, by studying the UAV's trajectory with its initial and final locations in or out of the jamming-free area,
we can provide more perspectives and insights for UAV trajectory design when it is assigned a mission of monitoring.
In fact, this is why we consider a more general problem formulation and the proposed solution is applicable to 
different scenarios, wherever the suspicious link is located.
\end{remark}
}

\subsection{Extension to Non-LoS Channels}
If the suspicious transmission and legitimate monitoring links are located in an urban area,
the channel between the suspicious source and destination experiences Rayleigh fading,
which can be modeled as \cite{guangchi}
\begin{equation}\label{rayleigh}
h_0^t = \beta_0 \xi_t d^{-\kappa}, ~\forall t
\end{equation}
where $\xi_t$ is an exponentially distributed random variable with unit mean,
and $\kappa \ge 2$ is the path loss exponent.
The UAV-GN links can be formulated by considering the probabilities of both LoS and non-LoS (NLoS) channels,
where the LoS probability at each time slot for the S-U $(j=1)$ or U-D $(j=2)$  link $p_{LoS,j}^t$ is given by \cite{zyong}
\begin{equation}
p_{LoS,j}^t = \frac{1}{1+C\exp{(-D[\theta_j^t-C])}}, ~\forall t.
\end{equation}
Here the values of $C$ and $D$ is reliant on the propagation environment,
and $\theta_j^t = \frac{180}{\pi} \sin^{-1}(H/d_j^t)$
is the elevation angle in degree, which is closely related to the UAV's distance from the source node $d_1^t$
or the destination node $d_2^t$.
Thereby, the channel power gains of the UAV-GN links are given by \cite{zyong}
\begin{equation}\label{nonlos}
h_j^t = p_{LoS,j}^t \beta_0 {d_j^t}^{-\kappa} + (1-p_{LoS,j}^t) \zeta \beta_0 {d_j^t}^{-\kappa},~\forall t
\end{equation}
where $\zeta<1$ is the extra reduction factor for the NLoS channel.

The Rayleigh fading in \eqref{rayleigh} does not affect the original problem \eqref{p1},
while the NLoS component in \eqref{nonlos} renders problem \eqref{p1} hardly tractable for existing solvers.
To deal with it, we consider the case when $\kappa = 2$; then \eqref{nonlos} can be approximated by \cite{moza17}
\begin{equation}\label{quadratic}
h_j^t \approx \eta_1 {d_j^t}^{-2} +\eta_2, ~\forall t
\end{equation}
where $\eta_1$ and $\eta_2$ are two coefficients relying on the UAV altitude.
Using the expressions in \eqref{quadratic}, the objective function and constraints for the NLoS scenario
are generally in the same form as those in the original problem \eqref{p1}.
Similar to problem \eqref{p2}, the variables in the NLoS scenario can be separated into three blocks, namely,
$\{P_j^t, x_t, y_t\}, \{u_t\}$, and $\{w_t\}$, due to the product of $P_j^t u_t w_t$ invited in constraints \eqref{p24}
by the NLoS component.
The NLoS problem is convex regarding each block of variables when the other two blocks are fixed,
and can thus be solved by the proposed block coordinate descent approach.
Note that due to the approximation in \eqref{quadratic}, only a sub-optimal solution can be obtained.

{\blue
\subsection{Generalization to eavesdropping non-outage events}
In this section, we propose a stochastic model for the eavesdropping system by considering Rayleigh fading for
the suspicious S-D link, i.e., $h_0^t = \beta_0 \xi_t d^{-\kappa}, ~\forall t$ [cf. Eq. (14)].
The UAV channels are all LoS, and successful eavesdropping is not required within each time slot anymore.
Instead, we impose a constraint of non-outage probability to guarantee that the total successful eavesdropping events
satisfy a certain threshold over time.
We introduce the following indicator function $I_t, \forall t$ to denote the successful eavesdropping event of the UAV:
\begin{equation}\label{eq.it}
I_t = \left\{
\begin{aligned}
&1, ~\text{if} ~ \gamma_U^t \geq \gamma_D^t \\
&0, ~\text{otherwise}\\
\end{aligned}\right.
\end{equation}
where $I_t = 1$ and $I_t=0$ indicate eavesdropping non-outage and outage events, respectively.

The original problem is extended to the following form.
\begin{subequations}\label{non-out}
\begin{align}
& \min_{\{P_j^t\}, \{x_t, y_t\}} \sum_{t \in {\cal T}} P_j^t \delta \label{non1}\\
&\text {s.t.} ~\sum_t I_t \ge p_{\text {non}}T_w  \label{non2}\\
&\qquad \eqref{p13}-\eqref{p16} \notag
\end{align}
\end{subequations}
where $p_{\text {non}} \in [0,1]$ is the eavesdropping non-outage probability and constraint \eqref{non2} guarantees that
at least $100p_{\text {non}}\%$ of the total eavesdropping performances are successfully operated.
Constraint \eqref{non2} is actually a relaxed (or generalized) version of constraints \eqref{p12},
which can also take the form of non-outage probability:
\begin{equation}
\mathbb{P}(\gamma_U^1 \ge \gamma_D^1, \ldots, \gamma_U^{T_w} \ge \gamma_D^{T_w}) \ge p_{\text {non}}.
\end{equation}

When $p_{\text {non}} =1$, problem \eqref{non-out} specializes to the original problem \eqref{p1}.
On the other hand, if $p_{\text {non}} =0$, jamming is not needed at all and the optimal value of the objective function
$\sum_{t} P_j^t \delta$ is zero.
In this case, problem \eqref{non-out} reduces to the feasibility problem of finding a trajectory constrained by the UAV's
maximum speed with $P_j^t = 0, \forall t$.
At optimality, jamming signals will be suppressed for at most $100(1-p_{\text {non}})\%$ of the $T_w$ time slots
with worse S-U channels (or higher jamming power consumptions),
and they will be sent, if necessary, in time slots with better S-U channels.
As problem \eqref{non-out} is a relaxed one of the original problem \eqref{p1},
the optimal UAV trajectory for \eqref{p1} is also an optimal one for \eqref{non-out},
and the optimal jamming energy in \eqref{p1} serves as an upper bound for that in \eqref{non-out}.
Problem \eqref{non-out} can be solved by first solving \eqref{p1},
then ranking the values of $\{P_j^{t*}\}_t$ from large to small and setting the top $100(1-p_{\text {non}})\%$ to zero.

\subsection{Extension to two suspicious links}
In this section, we extend the original problem (7) to include two suspicious links for the UAV to monitor simultaneously.
The second pair of suspicious ground source and destination nodes, S$_2$ and D$_2$,
are located at $(0, s_2)$ and $(d, s_2)$, respectively, where $s_2$ is the given y-coordinate of the nodes.
All communication links are assumed to be LoS for simplicity.
We assume that the UAV has three antennas with one of them for jamming and the other two for monitoring each link.
Note that in this case, jamming signal is sent to both links as long as eavesdropping is unsuccessful over one of the links.

The channel power gain for the S$_2$-D$_2$ link is the same as the S-D link, $h_0 = \beta_0 / d^2$.
The channel power gains for the S$_2$-U and U-D$_2$ links are given by
\begin{subequations}
\begin{align}
&h_{21}^t = \frac{\beta_0}{x_t^2 + (s_2-y_t)^2 + H^2}, ~~\forall t \\
&h_{22}^t= \frac{\beta_0}{(d-x_t)^2 + (s_2-y_t)^2 + H^2}, ~~\forall t.
\end{align}
\end{subequations}
The SINR (or SNR) of the S$_2$-D$_2$ and S$_2$-U links are given by
\begin{subequations}
\begin{align}
&\gamma_{D_2}^t = \frac{h_{0}P_x^t}{h_{22}^tP_j^t+\sigma^2}, ~~\forall t \\
&\gamma_{U_2}^t = \frac{h_{21}^tP_x^t}{\sigma^2}, ~~\forall t.
\end{align}
\end{subequations}
Successful eavesdropping requires that $\gamma_{U_2}^t \ge \gamma_{D_2}^t$, and $\gamma_{U}^t \ge \gamma_{D}^t$
for both links.
The new problem of interest can be formulated as
\begin{subequations}\label{ptwo}
\begin{align}
& \min_{\{P_j^t\}, \{x_t, y_t\}} \sum_{t \in {\cal T}} P_j^t \delta \label{ptwo1}\\
&\text {s.t.} ~~\frac{h_{0}P_x^t}{h_{2}^tP_j^t+\sigma^2} \leq \frac{h_{1}^tP_x^t}{\sigma^2}, ~\forall t  \label{ptwo2}\\
&\qquad \frac{h_{0}P_x^t}{h_{22}^tP_j^t+\sigma^2} \leq \frac{h_{21}^tP_x^t}{\sigma^2}, ~\forall t  \label{ptwo3}\\
&\qquad \eqref{p13}-\eqref{p16}.\notag
\end{align}
\end{subequations}

Problem \eqref{ptwo} can be solved by following the same procedure summarized in Algorithm 1.
The jamming-free area of the S$_2$-D$_2$ link is ${\cal A}_2:= \{(x_t, y_t)| \sqrt{x_t^2 + (s_2-y_t)^2 + H^2} \le d, \forall t\}$.
When $|s_2| \le 2d$, the common jamming-free area, i.e., the jamming-free area for problem \eqref{ptwo}
is the intersection of ${\cal A}$ and ${\cal A}_2$,
which is essentially the intersection of two circles centered at $(0,0)$ and $(0,s_2)$, respectively, both with a radius of $d$.
When $|s_2| > 2d$, the common jamming-free area does not exist as ${\cal A} \mathop{\cap}  {\cal A}_2 = \varnothing$.

It is worth noting that the problem of two suspicious links can be further extended to address multiple suspicious links with different separating distances, or aerial (rather than ground) suspicious nodes with 3D optimization of the UAV trajectory. 
}

In a nutshell, we address the problem of jamming energy minimization for a UAV-enabled monitoring system based on
the assumption of sufficient power supply.
We provide useful insights on the UAV trajectory design and reveal its impact on the jamming policy.
However, in practice, such a trajectory design could result in a great cost (and waste) of propulsion power.
Furthermore, it is not possible for untethered UAVs to possess infinite power supply during flight.
Motivated by this, we next investigate the energy optimization based on a more practical setting,
by considering finite power supply and propulsion power consumption at the UAV.

\section{Energy Management for Solar-Powered UAV}\label{sec.energy}

Compared with cables, laser-beams, and on-board batteries, the solar-powered UAVs enjoy a high flexibility and a long flight endurance
in practical deployment.
Apart from communication and jamming power, the UAV consumes additional propulsion power
to maintain airborne and support its movement.
Energy-efficient operation of the UAV needs to be achieved by considering propulsion energy management in system design
\cite{zyong}.

Suppose that the UAV has a solar panel to harvest energy and an on-site battery to save energy.
The UAV's battery is initially charged with $E_0$ amount of energy,
and that it can consume $\vartheta$ portion of $E_0$ during the entire working period,
and save the $(1-\vartheta)$ portion for emergency during landing (to a prescribed platform or home base).
The UAV can fly horizontally and adjust its positions dynamically to enhance the eavesdropping performance.
We pursue the optimal trajectory design and energy management scheme of the UAV by minimizing the total
jamming and propulsion energy consumption for the solar-powered rotary-wing UAV enabled monitoring system.

\begin{figure}[t]
\centering
\includegraphics[width=0.6\textwidth]{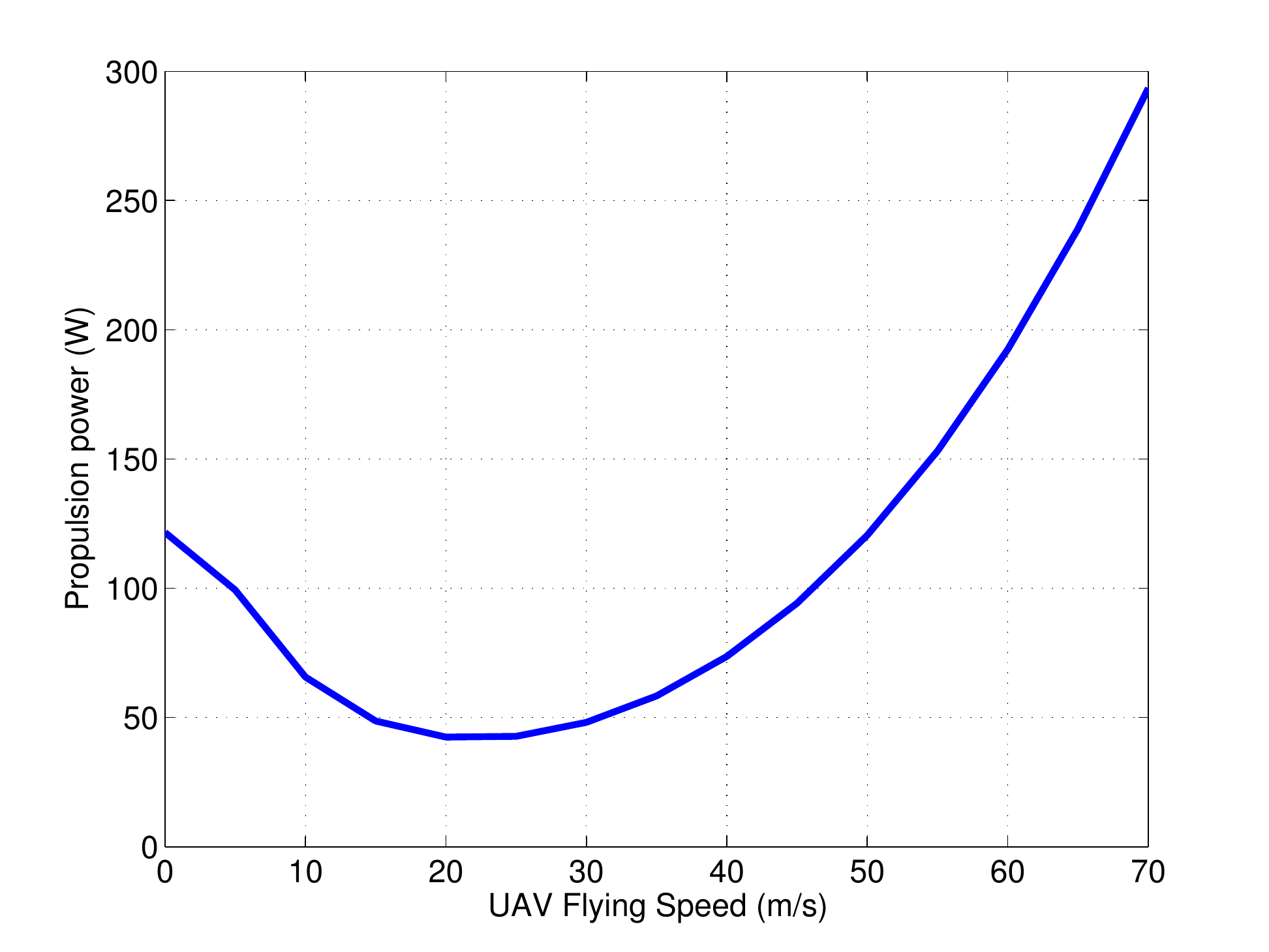}
\caption{Propulsion power versus UAV speed.}
\label{pm}
\end{figure}

\begin{table}[t]
\caption{Parameters for propulsion power \cite{zyong}}
\begin{center}
\begin{tabular}{|c|c|}
\hline
\text{UAV weight}                    &2 kg \\ \hline
\text{Blade profile power and induced power, $P_0$, $P_1$}               &3.4 W, 118 W \\ \hline
\text{Rotor solidity and disc area, $s$, $A$}                       &0.03, 0.28 m$^2$\\ \hline
\text{Tip speed of the rotor blade, $U_{tip}$}     &60 m/s \\ \hline
\text{Mean rotor induce velocity, $v_0$ }     &5.4 m/s\\ \hline
\text{Atmospheric density and fuselage drag ratio, $\rho$, $d_f$}                       &1.225 kg/m$^3$, 0.3 \\ \hline
\end{tabular}
\end{center}
\label{tab.pm}
\end{table}
\subsection{UAV Propulsion Power Model}
Besides transmit (jamming) power, the communication related power includes also that for communication circuitry,
information receiving and decoding, signal processing, etc.
For simplicity, we suppose that such communication connected power is a constant, represented by $P_c$ in watt (W) \cite{szhang, hu20}.
The propulsion power, which typically depends on the UAV speed, is essential to support the UAV's
hovering and moving activities.
For a rotary-wing UAV with speed $V_t$,
the propulsion power at time slot $t$, denoted by $P_m^t$, is given by \cite{zyong}
\begin{equation}\label{speed}
\begin{aligned}
P_m^t = & P_0 \left(1+\frac{3V_t^2}{U_{tip}^2}\right) + P_1\left(\sqrt{1+\frac{V_t^4}{4v_0^4}}-\frac{V_t^2}{2 v_0^2}\right)^{\frac{1}{2}}\\
&+ \frac{1}{2} d_f\rho s AV_t^3
\end{aligned}
\end{equation}
where $P_0$ and $P_1$ have fixed values and stand for the {\it blade profile power} and {\it induced power} under hovering mode,
respectively, $U_{tip}$ is the tip speed of the rotor blade, $v_0$ denotes the average rotor induced velocity in hover,
$d_f$ and $s$ represent the fuselage drag ratio and rotor solidity,
and $\rho$ and $A$ are the atmospheric density and rotor disc area, respectively.
When $V_t=0$, \eqref{speed} corresponds to the power consumption of the hovering state.
Fig. \ref{pm} depicts the typical curve of $P_m^t$ versus $V_t$.
The parameters are set according to Table \ref{tab.pm} \cite{zyong}.
It is revealed by Fig. \ref{pm} that the UAV speed achieving the least power consumption, i.e., about $41.84$ W, 
happens at approximately $V_e = 22.36$~m/s.

\begin{figure}[t] 
\centering
\includegraphics[width=0.6\textwidth]{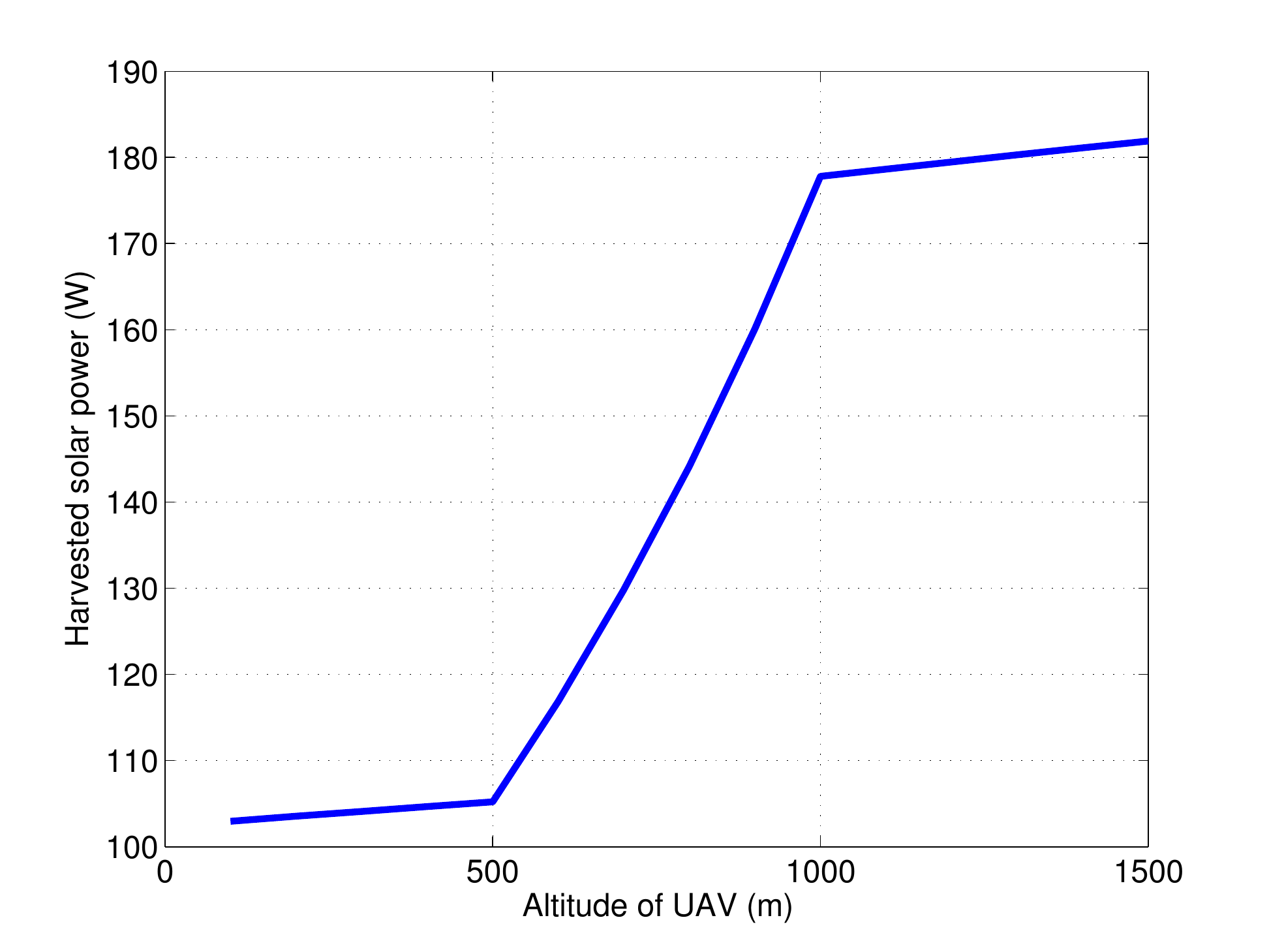}
\caption{Harvested solar power versus UAV altitude.}
\label{solar}
\end{figure}
\begin{table}[t]
\caption{Parameters for solar power \cite{kokh04}}
\begin{center}
\begin{tabular}{|c|c|}
\hline
\text{Atmospheric transmittance, $\alpha$, $\beta_1$}                        &0.8978, 0.2804 \\ \hline
\text{Interception factor of clouds, $\beta_c$}                                 &0.01 \\ \hline
\text{Mean radiant power and scaling altitude, $F$, $\Delta$}               &1367 W/m$^2$, 8000 m\\ \hline
\text{Efficiency and size of solar panel, $\eta$, $S$}                             &0.4, 0.5 m$^2$ \\ \hline
\text{Altitude of cloud, $L$, $U$}                                                           &500 m, 1000 m \\ \hline
\end{tabular}
\end{center}
\label{tab.solar}
\end{table}

We suppose that within each slot $t$, the UAV maintains a constant speed, which is given by
\begin{equation}\label{vt}
V_t = \sqrt{(x_t-x_{t-1})^2 + (y_t - y_{t-1})^2}/ \delta, ~\forall t.
\end{equation}
By substituting \eqref{vt} into \eqref{speed}, we find that the first and the third terms of \eqref{speed}
are jointly convex with respect to $(x_t,x_{t-1}, y_t, y_{t-1})$,
whereas the second term is neither convex nor concave.

\subsection{Solar Power Model}
Generally, the amount of the harvested solar power depends on the atmospheric transmittance and clouds.
As higher altitude results in higher solar intensity, the atmospheric transmittance increases with the altitude,
which can be empirically approximated by the following equation at altitude $H$  \cite{lee17}
\begin{equation}
\phi(H)=\alpha - \beta_1 e^{-H/\Delta}
\end{equation}
where $\alpha$ is the largest possible amount of the atmospheric transmittance,
$\beta_1$ is the extinguishing coefficient of the air,
and $\Delta$ is the scaling altitude.
On the other hand, the solar strength is diminished by cloud.
The reduction of sun light traveling through a cloud can be formulated as \cite{derrick, kokh04}
\begin{equation}
\psi(d_c)=e^{-\beta_c d_c}
\end{equation}
where $\beta_c \geq 0$ is the interception factor of the cloud,
and $d_c$ represents the spacial length that the sunlight travels through the cloud.
Overall, the electric generation power of the solar panels at height $H$ is given by \cite{lee17, kokh04}
\begin{equation}\label{eq.solar}
E(H)=\left\{
\begin{aligned}
&\eta SF\phi(H)\psi(0), H \geq U \\
&\eta SF\phi(H)\psi(U-H), L \leq H < U \\
&\eta SF\phi(H)\psi(U-L), H < L \\
\end{aligned} \right.
\end{equation}
where $\eta \in (0,1)$ and $S$ (in m$^2$) denote the efficiency and size of the solar panel, respectively.
Constant $F$ is the mean radiant power on the ground,
while $U$ and $L$ are the heights of the upper and lower limits of the cloud, respectively.
Fig.~\ref{solar} illustrates the influence of UAV's altitude on the harvested solar power.
The setting of the corresponding parameters are listed in Table \ref{tab.solar} \cite{kokh04}.


\subsection{Problem Formulation}
We aim to design the joint energy management and trajectory planning scheme for the solar-powered UAV aided eavesdropping system by minimizing its total energy consumption, including the jamming energy and the propulsion energy.
Since the UAV flies at a fixed altitude, we can simply use $E$ to denote the amount of solar power instead of $E(H)$.
The problem is formulated as
\begin{subequations}\label{p3}
\begin{align}
&\min_{\{P_j^t, x_t, y_t \}} \sum_{t \in {\cal T}} (P_j^t+P_m^t) \delta \label{p31}\\
&\text {s.t.} ~\sum_{i=1}^t (P_j^t + P_m^t +P_c) \delta \leq \sum_{i=1}^t E + \vartheta E_0, ~\forall t \label{p32}\\
&\eqref{p12} - \eqref{p16}. \notag
\end{align}
\end{subequations}
Constraint \eqref{p32} is the energy-harvesting causality constraint, which is imposed to bound the total consumed energy up to the
current time slot not to exceed the harvested energy plus the battery capacity.

The minimum level of the initially stored energy $\underline E_0$ is chosen such that the UAV can finish the eavesdropping mission without harvested solar energy,
following the shortest trajectory at a constant speed.
In particular, given the line segment connecting its horizontal initial and final locations $(x_0, y_0)$ and $(x_T, y_T)$,
the UAV travels at a fixed speed $\overline V = \sqrt{(x_T-x_0)^2+(y_T-y_0)^2}/T$.
With $\overline V$, we can obtain the total propulsion energy $P_m$ and the UAV's coordinates $(x_t, y_t)$ at each time slot.
Based on the coordinates, we can further calculate its jamming power $P_j^t$ according to \eqref{p12} per time slot.
Then, we can readily obtain the value of $\underline E_0 = P_m + \sum_t (P_j^t+P_c)\delta$.

{\blue
\begin{remark}\textit {(3D UAV trajectory design with altitude optimization):}
3D UAV trajectory design can be pursued by including the UAV altitude as an optimization variable $H_t, \forall t$.
Considering problem (7) and the S-U channel condition
$h_1^t = \frac{\beta_0}{{d_{1}^{t}}^2}=\frac{\beta_0}{x_t^2+y_t^2+H_t^2}, \forall t$, 
the optimal altitude for the UAV is the lowest height within the regulated range that it can stay,
since the UAV enjoys the best channel condition in this way and there is no performance gain by increasing its altitude.

On the other hand, considering the model of harvested solar power in Section IV-B [cf. \eqref{eq.solar}],
a tradeoff can be observed between the UAV channel conditions and the amount of harvested energy [cf. \eqref{p3}].
The UAV has to decide at each time slot whether to fly lower or higher to strike a balance between achieving better eavesdropping performance and harvesting more energy.
With the UAV altitude included as an optimization variable $H_t, \forall t$,
constraints \eqref{p32} become non-convex as the amount of harvested energy $E_t$ is altitude-dependent and time-varying.
It is difficult to convert \eqref{p32} to convex constraints due to the complicated expression of $E_t$,
thus rendering the new problem hardly tractable for existing solvers.
Furthermore, to the best of our knowledge, there is not a general model to capture the power consumptions incurred by both horizontal and vertical movements of the UAV,
which in turn makes it difficult to pursue a joint 3D UAV trajectory design and power allocation.
It will be an interesting direction to pursue in our future works with altitude optimization.
\end{remark}
}

\subsection{SCA-based Convexification and Solution}
The problem \eqref{p3} is not convex since it consists the non-convex term
$\left(\sqrt{1+\frac{V_t^4}{4v_0^4}}-\frac{V_t^2}{2 v_0^2}\right)^{\frac{1}{2}}$ in $P_m^t$,
and the non-convex constraints \eqref{p12}.
The latter can be handled by leveraging the same method as in Section \ref{sec:jam}.
With slack variables $\{u_t, w_t, \forall t\}$, \eqref{p12} can be replaced with the constraints \eqref{p22}-\eqref{p25}.

To tackle the non-convexity with $P_m^t$, we first bring in slack variables $\{q_t \geq 0\}$ such that
\begin{equation}
q_t^2 = \sqrt{1+\frac{V_t^4}{4v_0^4}}-\frac{V_t^2}{2 v_0^2}, ~\forall t
\end{equation}
which is equivalent to
\begin{equation}\label{mu}
\frac{1}{q_t^2} = q_t^2 + \frac{V_t^2}{v_0^2}, ~\forall t.
\end{equation}
The second term of \eqref{speed} can thus be substituted by the linear component $P_1 q_t$,
with the additional constraints \eqref{mu}.
For the purpose of exposition, we now integrate the expression for $V_t$ in \eqref{vt} and let
\begin{equation}\label{newpm}
\begin{aligned}
{\tilde P}_m^t := & P_0 + \frac{3P_0}{U_{tip}^2 \delta^2}\left[(x_t - x_{t-1})^2 + (y_t - y_{t-1})^2\right] + P_1 q_t \\
& + \frac{d_f}{2\delta^3} \rho s A\left[(x_t - x_{t-1})^2 + (y_t - y_{t-1})^2\right]^{3/2}, ~\forall t.
\end{aligned}
\end{equation}
We can see that ${\tilde P}_m^t$ is now jointly convex with respect to $(x_t,x_{t-1}, y_t, y_{t-1}, q_t)$.
With such a manipulation, problem \eqref{p3} can be written as
\begin{subequations}\label{p3.1}
\begin{align}
& \min_{\substack{\{P_j^t, q_t,u_t \} \\ \{x_t, y_t,w_t\}}} \sum_{t \in {\cal T}} ( P_j^t+{\tilde P}_m^t ) \delta \label{p3.11}\\
&\text {s.t.}~ \sum_{i=1}^t ( P_j^t + \tilde P_m^t + P_c )\delta \leq \sum_{i=1}^t E + \vartheta E_0, ~\forall t \label{p3.12} \\
&\frac{1}{q_t^2} \leq q_t^2 + \frac{(x_t - x_{t-1})^2 + (y_t - y_{t-1})^2}{\hat v_0^2 }, ~\forall t \label{p3.14}\\
&\eqref{p13} - \eqref{p16}, \eqref{p22} - \eqref{p25} \notag
\end{align}
\end{subequations}
where $\hat v_0^2=v_0^2 \delta^2$.

Note that constraints \eqref{p3.14} are obtained from \eqref{mu} by replacing the equations with inequalities.
Yet, equivalence still holds between problems \eqref{p3} and \eqref{p3.1}.
To examine this, we assume that if any of the constraints in \eqref{p3.14} is met with strict inequality when
achieving optimality for problem \eqref{p3.1},
we can decrease the related value of variable $q_t$ to enable constraint \eqref{p3.14} met with equality,
while reducing the total energy consumption (objective value) at the same time.
Therefore, all constraints in \eqref{p3.14} are met with equality at optimality.
The same equivalence also holds for constraints \eqref{p22} and \eqref{p23} as explained in Section \ref{sec:jam}.
Hence, problems \eqref{p3} and \eqref{p3.1} are equivalent.

Problem \eqref{p3.1} is still non-convex since it consists the non-convex constraints in \eqref{p3.14}.
However, it can be tackled with the successive convex approximation (SCA) method
by calculating the global lower bounds at a given local point.
In particular, for \eqref{p3.14}, the left-hand-side (LHS) is a convex function in $q_t$, and the right-hand-side (RHS) is a jointly convex
function regarding $q_t$ and $(x_t,x_{t-1}, y_t, y_{t-1})$.
Since the first-order Taylor expansion serves as the global lower bound of a convex function,
we can obtain the following inequality for the RHS of \eqref{p3.14}
\begin{equation}\label{eq.q}
\begin{aligned}
&q_t^2 + \frac{(x_t - x_{t-1})^2 + (y_t - y_{t-1})^2}{\hat v_0^2} \geq q_t^{(l)2} + 2q_t^{(l)}(q_t - q_t^{(l)})  \\
&+\frac{2}{\hat v_0^2}[(x_t^{(l)} - x_{t-1}^{(l)})(x_t - x_{t-1})+(y_t^{(l)} - y_{t-1}^{(l)})(y_t - y_{t-1})]\\
&- \frac{1}{\hat v_0^2}[(x_t^{(l)} - x_{t-1}^{(l)})^2 + (y_t^{(l)} - y_{t-1}^{(l)})^2]
\end{aligned}
\end{equation}
where $q_t^{(l)}$, $x_t^{(l)}$, and $y_t^{(l)}$ are the present values of the corresponding variables at the $l$-th iteration, respectively.
\begin{algorithm}[t]
\caption{SCA-based Method for Problem \eqref{p3.2}}
\label{algo:sco}
\begin{algorithmic}[1]
\State {\bf Initialization:} Find an initially feasible solution $\{P_j^t(0), x_t(0), y_t(0), q_t(0),u_t(0), w_t(0)\}$ for Problem \eqref{p3.2}.
\For {$l$ = 0, 1, 2, ...}
\State Obtain the optimal solution of $\{P_j^t(l+1), q_t(l+1), u_t(l+1)\}$ with $\{q_t(l), x_t(l), y_t(l), w_t(l) \}$ fixed.
\State Compute the optimal solution of $\{x_t(l+1), y_t(l+1), w_t(l+1) \}$ with $\{P_j^t(l+1), q_t(l+1), u_t(l+1)\}$ fixed.
\State Update $l=l+1$.
\EndFor
\end{algorithmic}
\end{algorithm}
By substituting the non-convex constraints \eqref{p3.14} with its lower bound at the $l$-th iteration acquired by \eqref{eq.q},
we can establish the following optimization problem
\begin{subequations}\label{p3.2}
\begin{align}
& \min_{\substack{\{P_j^t, q_t,u_t \} \\ \{x_t, y_t,w_t\}}} \sum_{t \in {\cal T}} ( P_j^t+{\tilde P}_m^t ) \delta \label{p3.21}\\
&\text {s.t.}~ \frac{1}{q_t^2} \leq q_t^{(l)2} + 2q_t^{(l)}(q_t - q_t^{(l)}) 
+\frac{2}{\hat v_0^2}[(x_t^{(l)} - x_{t-1}^{(l)})(x_t - x_{t-1})+(y_t^{(l)} - y_{t-1}^{(l)})(y_t - y_{t-1})] \notag \\
& \qquad \qquad -\frac{1}{\hat v_0^2}[(x_t^{(l)} - x_{t-1}^{(l)})^2 + (y_t^{(l)} - y_{t-1}^{(l)})^2], ~\forall t \label{p3.22}\\
&q_t \geq 0, ~\forall t \label{p3.23} \\
&\eqref{p13} - \eqref{p16}, \eqref{p22} - \eqref{p25}, \eqref{p3.12}. \notag
\end{align}
\end{subequations}
It can be justified that problem \eqref{p3.2} is convex in $\{P_j^t, q_t,u_t \}$ for fixed $\{x_t, y_t,w_t\}$,
and it is convex in $\{x_t, y_t,w_t\}$ for fixed $\{P_j^t, q_t,u_t \}$.
Similarly, we can leverage the alternating optimization method to acquire the optimal values of one block of variables with the other fixed
iteratively.
The proposed algorithm is summarized in Algorithm \ref{algo:sco}.

{\blue
In the proposed algorithm, each subproblem is a convex program,
which can be efficiently tackled via classic convex optimization methodologies in polynomial time.
It is worth noting that because of the global lower bounds in \eqref{eq.q}, when the constraints of problem \eqref{p3.2} are fulfilled,
those for the original problem \eqref{p3.1} are also fulfilled; yet the reverse does not necessarily hold.
Thereby, the feasible region of \eqref{p3.2} is a subset of that for \eqref{p3.1},
and the optimal value of \eqref{p3.2} draws an upper limitation to that of \eqref{p3.1}.
By sequentially renewing the local point at each iteration through solving \eqref{p3.2}, our proposed approach is established
for the non-convex optimization problem \eqref{p3.1} or its original problem \eqref{p3}.
Through the similar statements in \cite{zyong} and \cite{zappone},
it is demonstrated that the proposed approach is ensured to converge to at least a solution
that fulfills the KKT conditions of problem \eqref{p3.1}.
A high-quality sub-optimal solution can therefore be obtained by our proposed algorithm
with a computational complexity of $\mathcal{O}(T_w^{3.5})$ at a fast convergence speed,
as will be corroborated by simulation results provided in Section V.
}

\begin{figure}[t]
\centering
\includegraphics[width=0.65\textwidth]{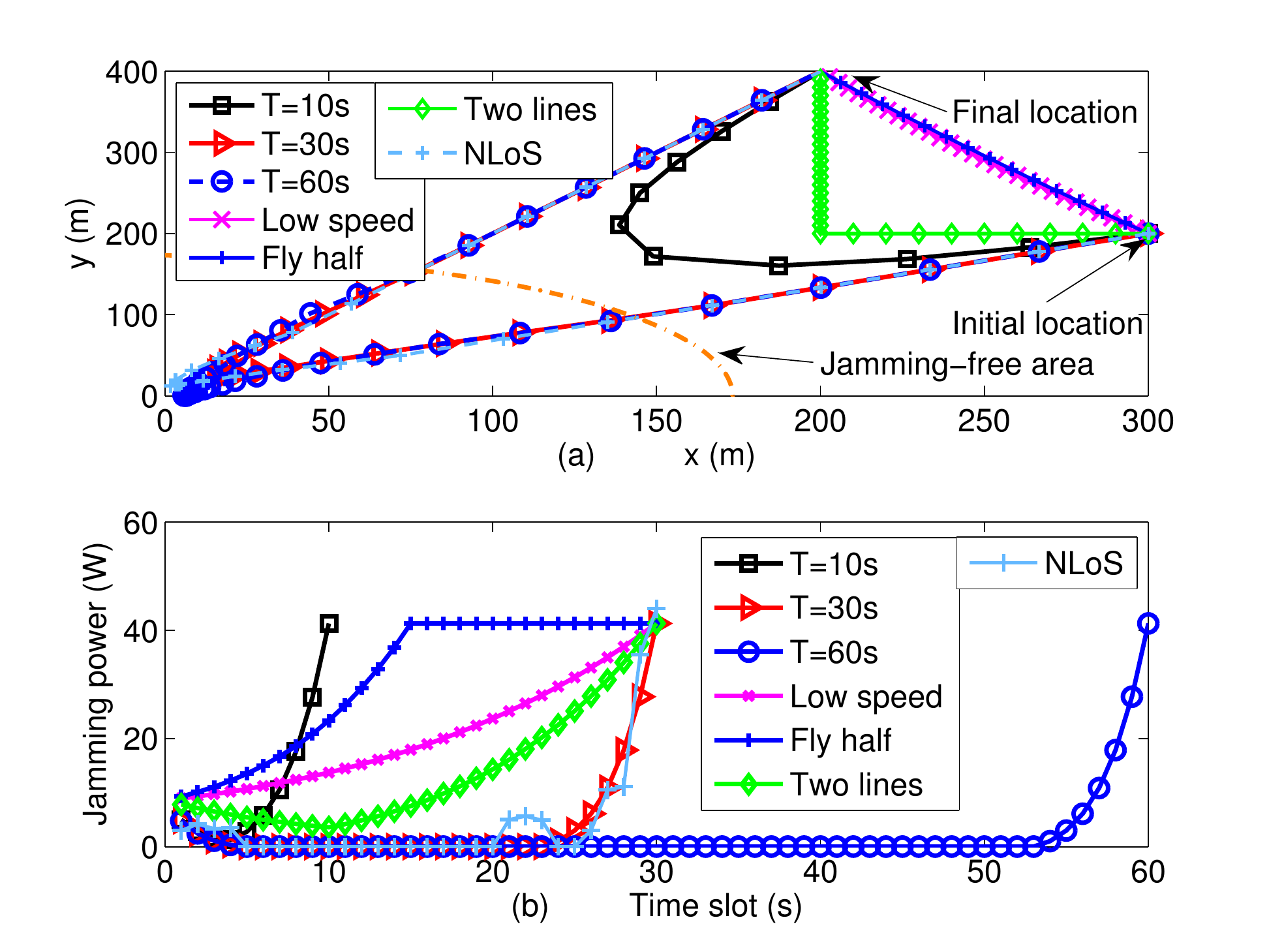}
\caption{UAV trajectory designs and jamming power allocations under NF scenario.}
\label{bothben}
\end{figure}




\section{Numerical Results}\label{sec.sim}
In this section, we provide numerical results for the proposed approaches.
The reference channel power $\beta_0$ is set as $10^{-12}$, the noise $\sigma$ is set as $-169$ dBm/Hz,
and the communication bandwidth is $10$ MHz.
The distance between S and D is $d = 200$ m, and the UAV flies at an altitude of $H = 100$ m.
The maximum horizontal speed of the UAV is set as $\tilde V_m = 40$ m/s.
The slot length is $\delta = 0.1$ s.
The original capacity of the battery $E_0$ is $7 \times 10^3$ J.
Parameters concerning the propulsion power and the harvested solar power are the same as in Tables \ref{tab.pm} and \ref{tab.solar}.
To evaluate the proposed optimal trajectory design and power allocation schemes,
we test three pairs of coordinates for the initial and final locations of the UAV.
The three test cases are:
1) JF (Jamming Free) scenario: both the initial and final locations are inside the jamming-free area of $\cal A$, namely,
$(x_0, y_0) = (-50~{\text m},-100~{\text m})$, and $(x_T, y_T) = (100~{\text m},140~{\text m})$;
2) IF (Initial jamming Free) scenario: the initial location is inside $\cal A$ and the final location is outside $\cal A$, namely,
$(x_0, y_0) = (-50~{\text m},0)$, and $(x_T, y_T) = (100~{\text m},350~{\text m})$;
and 3) NF (No jamming Free) scenario: both locations are outside $\cal A$, namely,
$(x_0, y_0) = (300~{\text m},200~{\text m})$, and $(x_T, y_T) = (200~{\text m},400~{\text m})$.
\begin{figure}[t]
\centering
\includegraphics[width=0.65\textwidth]{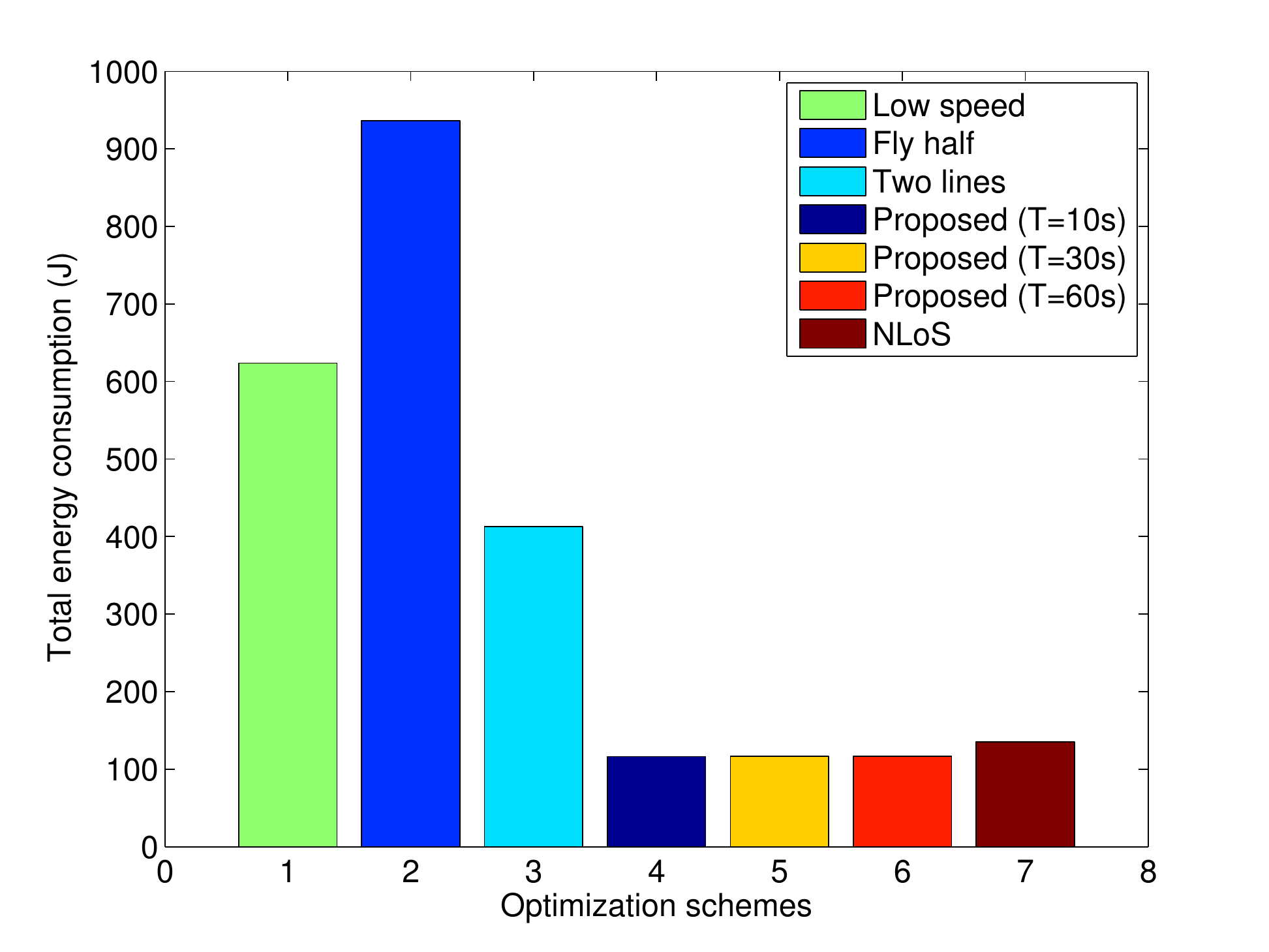}
\caption{Total jamming energy consumptions of the UAV under NF scenario.}
\label{totaljam}
\end{figure}
To further observe the UAV's behavior, we adopt three time horizons for each scenario,
namely, $T=10$ s, $30$ s and $60$ s.
Trajectory and power consumptions of the UAV are depicted every second.
Note that all pairs of coordinates are carefully selected such that at least one feasible trajectory can be found for the UAV
in the shortest time horizon.

Fig. \ref{bothben} depicts the UAV's trajectory designs (Fig. \ref{bothben}(a)) and jamming power allocations (Fig. \ref{bothben}(b))
for the simple system model in problem \eqref{p1} under the NF scenario.
The jamming-free area of $\cal A$ for the LoS links is illustrated by an orange dash-dot line.
It can be observed from Fig. \ref{bothben} that when both the initial and final locations are outside $\cal A$,
the UAV intends to fly towards $\cal A$ first, then travel to the final location.
With sufficient traveling time ($T=30$ s and $60$ s), the UAV first flies fast to $\cal A$,
then takes a detour at a very low speed inside $\cal A$,
and finally travels quickly to its final location.
During this process, the jamming power first decreases, then stays at zero, and finally increases quickly in the last few time slots.
This is consistent with the results in Lemma~\ref{lemma.time} and Proposition~\ref{prop.jam}.
We also include the performance of the UAV with NLoS links when $T=30$ s in Fig. \ref{bothben} (labeled as ``NLoS'').
It can be seen that the UAV's trajectory does not vary much under this scenario,
and that its jamming power does not change smoothly with its distance from the source due to the randomness invited by
the S-D link.
To further validate the advantage of trajectory design on energy reduction, we examine three baseline schemes of the UAV
under the NF scenario when $T=30$ s.
The first scheme is labeled as ``Low speed'', where the UAV travels straightly from the initial location to the final location
at a fixed speed ($7.46$ m/s).
The second scheme is labeled as ``Fly half'', where the UAV flies straightly to the final location at a constant speed ($14.91$ m/s)
during the first half of the period ($15$ s), and hovers at the destination for the rest of the period.
The third scheme is labeled as ``Two lines'', where the UAV first flies directly towards the point $(200~{\text m},200~{\text m})$,
then flies to the final location, following the trajectory of two line segments at a constant speed of $10$ m/s.

\begin{figure}[t]
\centering
\includegraphics[width=0.65\textwidth]{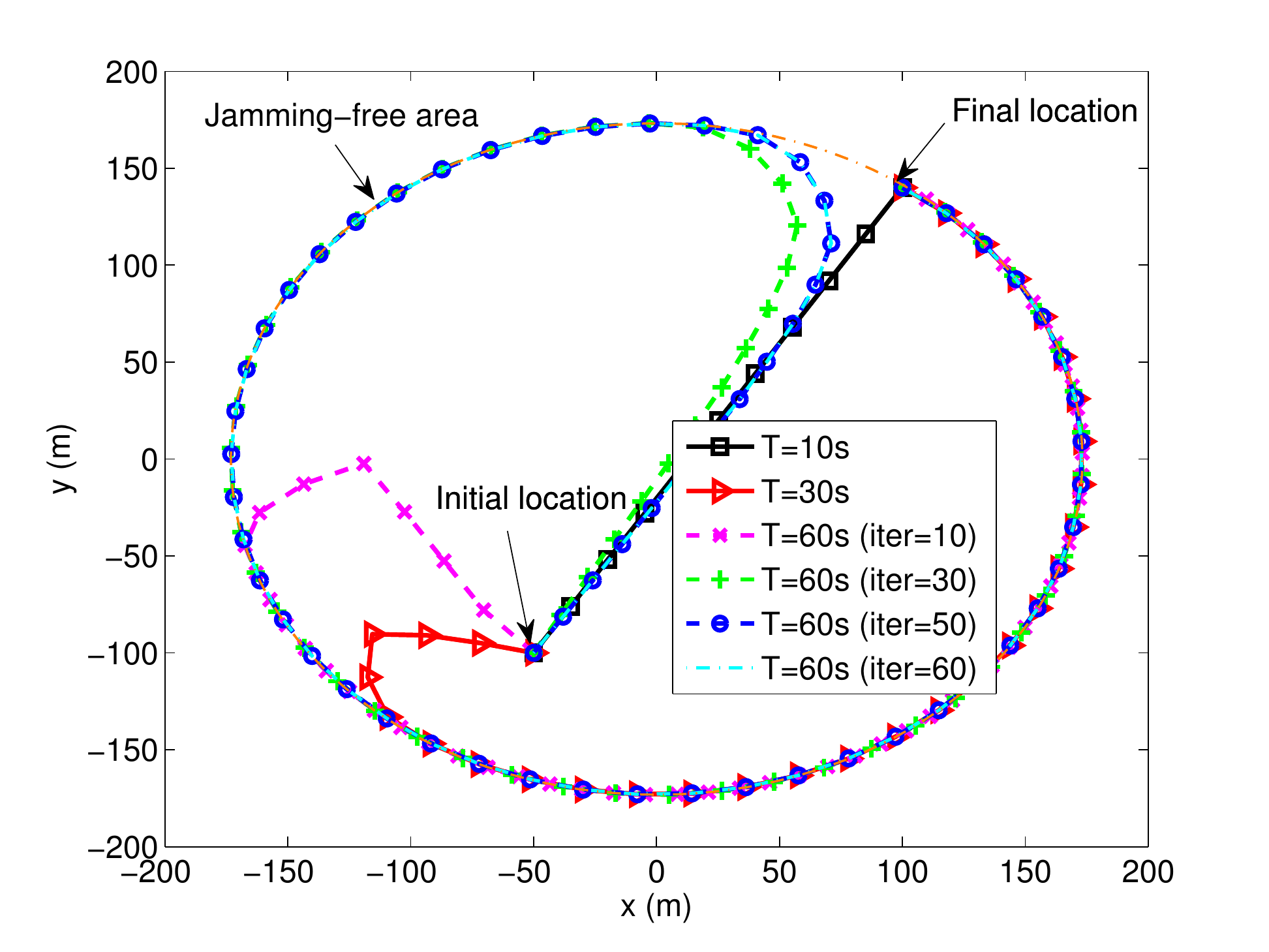}
\caption{UAV trajectory designs under JF scenario.}
\label{scojftj}
\end{figure}
\begin{figure}[th]
\centering
\includegraphics[width=0.65\textwidth]{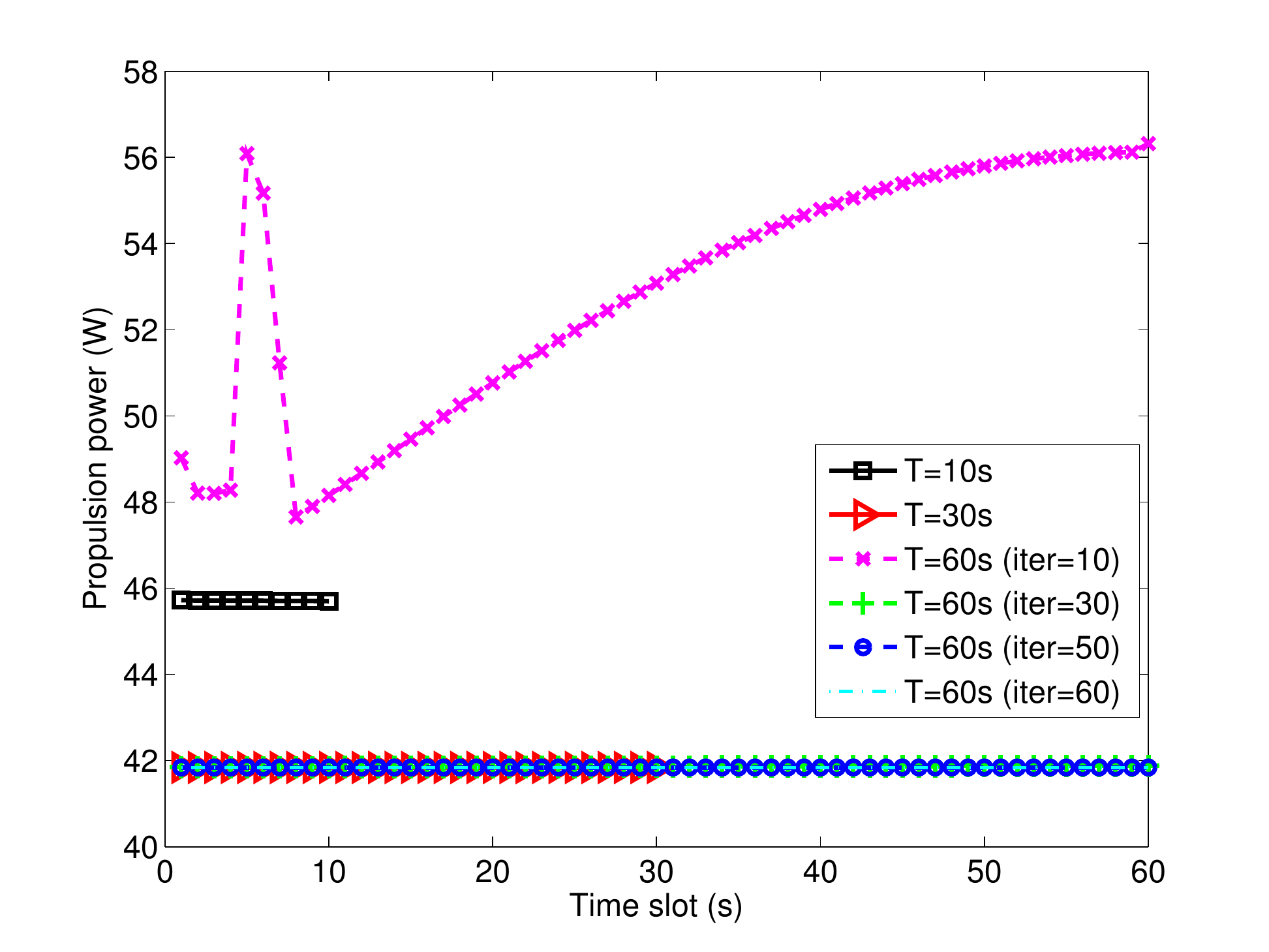}
\caption{UAV power allocations under JF scenario.}
\label{scojfp}
\end{figure}
\begin{figure}[th]
\centering
\includegraphics[width=0.65\textwidth]{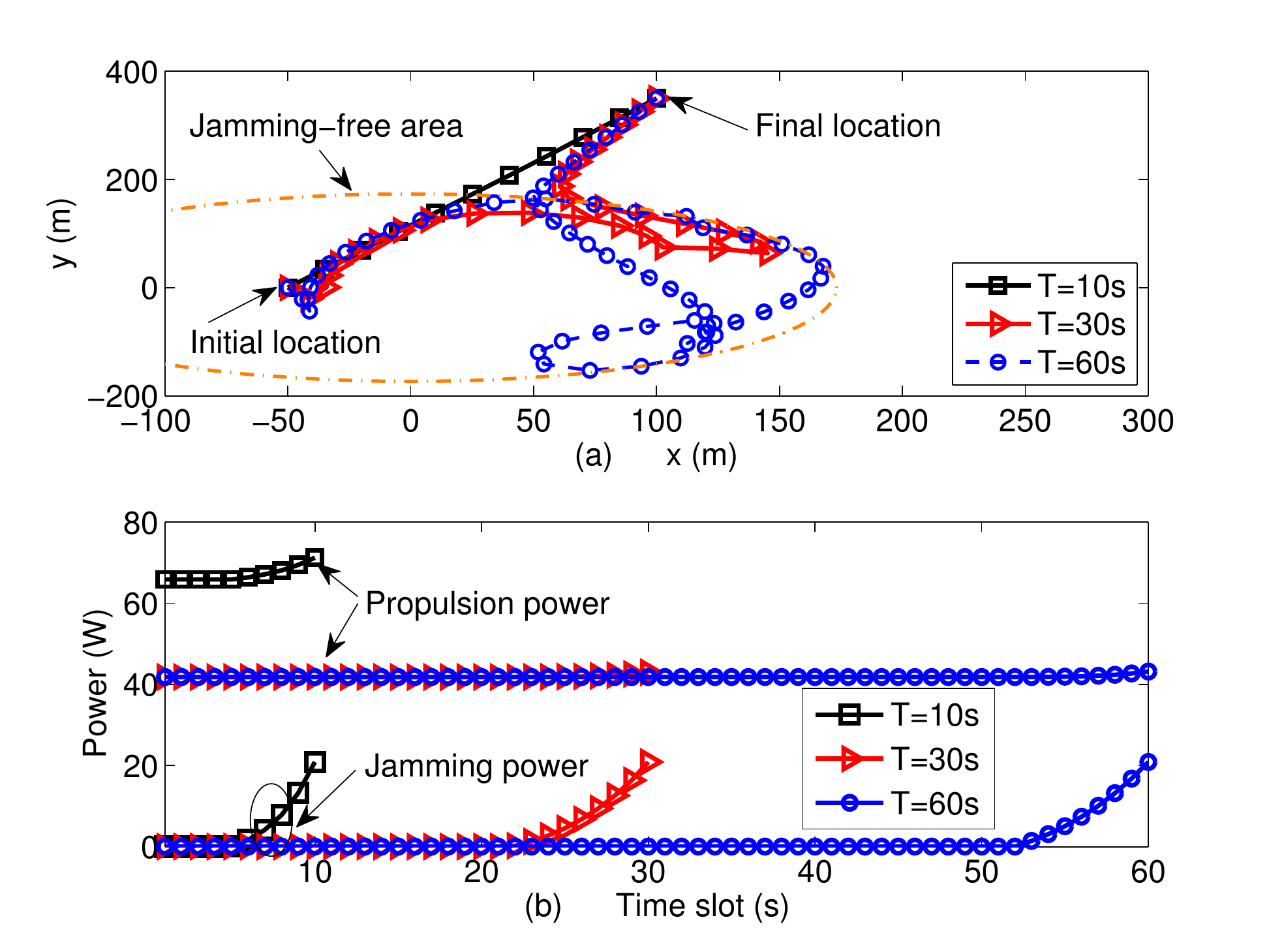}
\caption{UAV trajectory designs and power allocations under IF scenario.}
\label{scone}
\end{figure}

\begin{figure}[th]
\centering
\includegraphics[width=0.65\textwidth]{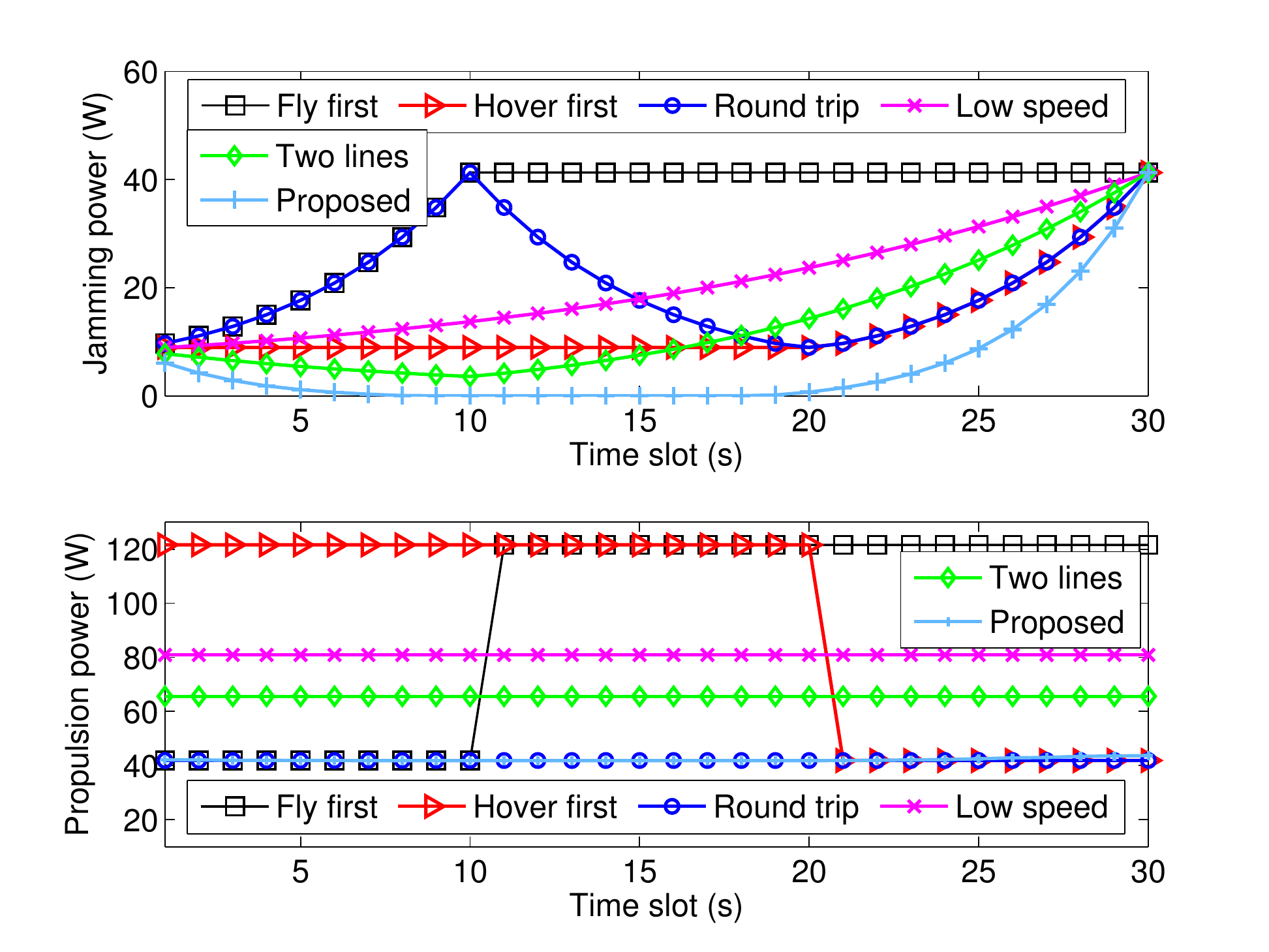}
\caption{UAV power allocations for baseline schemes under NF scenario.}
\label{ben4}
\end{figure}

Fig. \ref{totaljam} shows the total energy consumptions of the UAV under the NF scenario.
It is unveiled by Figs. \ref{bothben} and \ref{totaljam} that the UAV consumes significantly more jamming energy
without careful trajectory design.
{\blue The overall energy consumption of the ``Fly half'' scheme is almost ninefold of that of our proposed scheme,
since the UAV flies quickly to the destination and hovers there for a relatively long period.
As the destination is far from the source node, the longer the UAV stays there, the more jamming energy it consumes.}
The ``Two lines'' scheme is the most energy-efficient among the baseline schemes as it amounts to a simple optimization
of the trajectory.
Under the same parameter setting, the UAV consumes more energy with NLoS links than with LoS links,
since it experiences greater path loss with the former.
Note that for the proposed scheme, the total jamming energy does not increase with the length of the scheduling period,
since the UAV always spends the same duration outside $\cal A$.

\begin{figure}[t]
\centering
\includegraphics[width=0.65\textwidth]{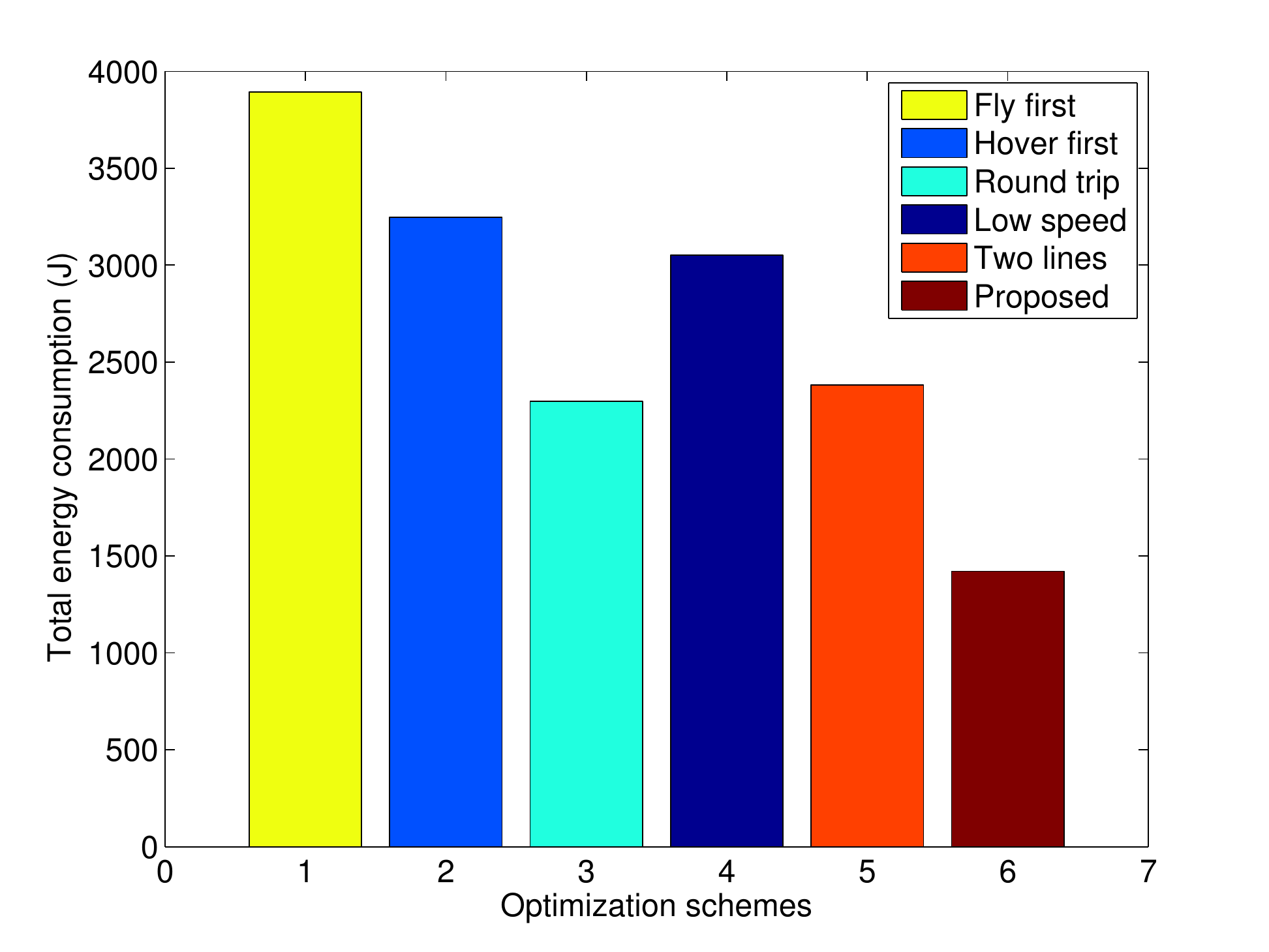}
\caption{Total energy consumptions of the UAV under NF scenario.}
\label{total}
\end{figure}




Figs. \ref{scojftj}--\ref{total} depict the trajectory designs and energy management schemes of the UAV
based on the system model proposed for problem \eqref{p3} in Section \ref{sec.energy},
where harvested solar energy, propulsion power and other circuit consumptions are considered.
Specifically, Figs. \ref{scojftj} and \ref{scojfp} demonstrate the convergence of the proposed approach for trajectory design
and propulsion power assignment of the UAV under the JF scenario.
Since the lines for the $50$-th and $60$-th iterations completely overlap,
a fast convergence within $50$ iterations can be readily observed in both figures.
The UAV travels either inside the jamming-free area of ${\cal A}$ or at the edge of ${\cal A}$ and incurs no jamming power at any time,
which corroborates Lemma \ref{lemma.free}.
When $T=10$ s, the UAV can only choose to finish the journey at the speed incurring as little power consumption as possible.
When the time horizon is sufficiently long ($T=30$ s and $60$ s), the UAV carefully designs its trajectory
and takes a detour inside ${\cal A}$ so that it can travel at the best energy-efficient speed $V_e$ all the way.


Furthermore, it is shown in Fig. \ref{scone} that for the IF scenario
when $T=10$ s, the UAV has to travel straightly from the initial location to the final location at a speed much faster than $V_e$,
thus leading to a significant amount of propulsion power consumption at each slot.
If there is surplus time,
the UAV first travels inside $\cal A$ at a constant speed of $V_e$, which minimizes the propulsion power consumption.
Then it flies to the final location, which is outside $\cal A$, in the last few time slots.
This trajectory design enables the UAV to stay inside $\cal A$ for as long as possible,
since the shorter time it stays outside $\cal A$, the less jamming energy it consumes.


To fully demonstrate the influence and merits of delicate trajectory design for the UAV,
we again compare with three baseline schemes where the UAV adopts different flying protocols for the same trajectory
as the ``Low speed'' scheme when $T=30$ s.
The first protocol is labelled as ``Fly first'', where the UAV flies to the final location at approximately $V_e = 22.36$ m/s
in the first $10$ s, then hovers at the destination for the rest $20$ s.
The second protocol is labelled as ``Hover first'', where the UAV hovers above the initial location in the first $20$ s,
then flies to the final location for the rest $10$ s at speed $V_e$.
The third protocol is labelled as ``Round trip'', where the UAV first takes a round trip between the initial and final locations,
then flies again to the destination, at speed $V_e$ during the flight period.
To facilitate comparison, we also include the ``Low speed'', ``Two lines'', and our proposed schemes under the NF scenario.
Figs. \ref{ben4} and \ref{total} depict the jamming and propulsion power allocations at each time slot,
and the total energy consumptions of the UAV, respectively.
{\blue Table \ref{tab.respective} lists the respective jamming and propulsion energy consumptions for different schemes.}
It can be readily seen from Fig. \ref{ben4} that the UAV needs to send jamming signals at every time slot
under the five baseline schemes,
as it is always traveling outside $\cal A$.
The propulsion power consumption for hovering triples that for traveling at speed $V_e$.
The ``Fly first'' scheme incurs the largest energy consumption for both jamming and propulsion,
due to its $20$ s hovering at the farthest point from the suspicious source.
It is further revealed in Fig. \ref{total} that the total energy consumption of the ``Round trip'' scheme
is the lowest among the baseline schemes,
since it adopts a simple trajectory design with the energy-efficient speed.
{\blue It is observed from Table \ref{tab.respective} that the jamming energy consumption is the highest
for the ``Fly first'' and ``Round trip'' schemes, while the propulsion energy consumption is the highest
for the ``Fly first'' and ``Hover first'' schemes.
Our proposed scheme consumes the least jamming energy and propulsion energy.}
Clearly, all of the baseline schemes consume more energy than our proposed scheme.
In a nutshell, the UAV suffers significant waste of energy without careful trajectory optimization.

\begin{table}[t]
\caption{Respective jamming and propulsion energy consumptions for different schemes under NF scenario.}
\begin{center}
\begin{tabular}{|c|c|c|}
\hline
\text{Optimization schemes}    &Jamming energy (J)  & Propulsion energy (J)   \\ \hline
\text{Fly first}               & 1042.9         & 2850.0     \\ \hline
\text{Hover first}          & 396.3          & 2850.0      \\ \hline
\text{Round trip}         & 1042.9         &1255.2       \\ \hline
\text{Low speed}        & 623.7        & 2427.5   \\ \hline
\text{Two lines}          &  412.8      & 1968.6      \\ \hline
\text{Proposed}          & 165.2       & 1255.2     \\ \hline

\end{tabular}
\end{center}
\label{tab.respective}
\end{table}

{\blue
\begin{remark}\textit {(Mitigating the interference on other links):}
When the suspicious link intentionally chooses to be located in a wild rural area to avoid surveillance by existing monitoring infrastructures,
there would be few communication links in the vicinity, and the interference caused by jamming could thereby be reduced to the minimum level, which is negligible.
In fact, it typically depends on the access scheme whether jamming suppresses communications of other links.
For instance, if the suspicious link occupies a certain frequency band all to itself,
the UAV is able to send exclusive jamming signals to it, which will not affect other links. 
On the other hand, if serious communication degradation is reported by legitimate users within the neighborhood,
the UAV can release the specific transmitted (and encrypted) information to these users
so that they can decode the jamming signals and will not be interfered.
Note that the maximum jamming power of $40$ W in Figs. \ref{bothben} and \ref{ben4} is the worst-case value tested in the simulation.
Yet in practice, the UAV is not usually that far away from the suspicious source
and does not incur such a high jamming power consumption.
\end{remark}
}

\section{Conclusion}\label{sec.con}
We addressed joint energy management and trajectory optimization for a rotary-wing UAV enabled legitimate
monitoring system.
Building on a judicious (re-)formulation, 
we leveraged the alternating optimization and successive convex approximation methodologies
to minimize the overall energy consumption of the UAV.
Efficient algorithms were developed to compute the locally optimal solution or 
at least a feasible solution fulfilling the KKT conditions.
We provided extensive numerical test results to justify the effectiveness of the proposed schemes.
The proposed framework also inspires new directions for future researches on security issues
in UAV-aided wireless networks such as wireless power transfer and/or mobile edge computing based ones,
especially with non-LoS channels and 3D trajectory planning.


\balance

\bibliographystyle{IEEEtran}
\bibliography{eaves_ref}

\end{document}